\documentclass[11pt,letterpaper]{article}
\usepackage{authblk}
\title{Small-Mass Asymptotics of Massive Point Vortex Dynamics in Bose--Einstein Condensates~I: Averaging and Normal Forms}
\author[1]{Tomoki Ohsawa}
\author[2]{Andrea Richaud} 
\author[3]{Roy Goodman}
\affil[1]{Department of Mathematical Sciences, The University of Texas at Dallas, 800 W Campbell Rd, Richardson, TX 75080-3021}
\affil[2]{Departament de Física, Universitat Politècnica de Catalunya, Campus Nord B4-B5, E-08034 Barcelona, Spain}
\affil[3]{Department of Mathematical Sciences, New Jersey Institute of Technology, University Heights, Newark, NJ 07102}
\date{\today}

\usepackage{graphicx,amsmath,amssymb,paralist,mathrsfs,braket}

\usepackage{caption,subcaption}
\usepackage[margin=1in, marginpar=.5in]{geometry}

\usepackage[numbers,sort&compress]{natbib}

\usepackage[colorlinks=false]{hyperref}

\usepackage{amsthm}
\theoremstyle{plain}
\newtheorem{theorem}{Theorem}
\newtheorem{corollary}[theorem]{Corollary}

\newtheorem{proposition}[theorem]{Proposition}
\newtheorem*{assumptions*}{Assumptions}

\theoremstyle{remark}
\newtheorem{remark}{Remark}

\usepackage[nameinlink]{cleveref}
\usepackage[most]{tcolorbox}

\newcommand{\od}[2]{\frac{d#1}{d#2}}
\newcommand{\pd}[2]{\frac{\partial #1}{\partial #2}}

\newcommand{\tpd}[2]{\partial #1/\partial #2}

\newcommand{\parentheses}[1]{\!\left(#1\right)}

\newcommand{\braces}[1]{\!\left\{#1\right\}}
\newcommand{\evalat}[2]{{\left.#1\right\rvert}_{#2}}

\newcommand{\Span}{\operatorname{span}} 

\newcommand{\diag}{\operatorname{diag}}

\newcommand{\norm}[1]{\left\|#1\right\|}
\newcommand{\abs}[1]{\left|#1\right|}
\newcommand{\DS}{\displaystyle}
\newcommand{\R}{\mathbb{R}}

\newcommand{\Z}{\mathbb{Z}}
\newcommand{\N}{\mathbb{N}}
\newcommand{\defeq}{\mathrel{\mathop:}=}
\newcommand{\eqdef}{=\mathrel{\mathop:}}
\newcommand{\ip}[2]{\left\langle#1,#2\right\rangle}

\newcommand{\Ran}{\operatorname{Ran}}
\newcommand{\Null}{\operatorname{Null}}
\newcommand{\Res}{\operatorname{Res}}
\newcommand{\Non}{\operatorname{Non}}
\newcommand{\eps}{\epsilon}
\newcommand{\rmi}{{\rm i}}
\newcommand{\cc}{\mathrm{c.c.}}

\renewcommand{\d}{\mathbf{d}}
\newcommand{\ins}[1]{{\bf i}_{#1}}
\newcommand{\PB}[2]{\left\{#1,#2\right\}}

\newcommand{\ad}{\operatorname{ad}}

\newcommand{\SO}{\mathsf{SO}}

\newcommand{\br}{\mathbf{r}}
\newcommand{\bp}{\mathbf{p}}
\newcommand{\bq}{\mathbf{q}}
\newcommand{\bv}{\mathbf{v}}
\newcommand{\bz}{\mathbf{z}}
\newcommand{\bw}{\mathbf{w}}

\newcommand{\bZ}{\mathbf{Z}}
\newcommand{\bW}{\mathbf{W}}
\newcommand{\hatbz}{\hat{\mathbf{z}}}
\newcommand{\bR}{\mathbf{R}}
\newcommand{\bP}{\mathbf{P}}
\newcommand{\cZ}{\mathcal{Z}}
\newcommand{\cW}{\mathcal{W}}
\newcommand{\cJ}{\mathcal{J}}

\newcommand{\cK}{\mathcal{K}}
\newcommand{\cR}{\mathcal{R}}
\newcommand{\cS}{\mathcal{S}}

\newcommand{\balpha}{\boldsymbol{\alpha}}
\newcommand{\bbeta}{\boldsymbol{\beta}}

\begin{document}

\maketitle

\begin{abstract}
 We perform an asymptotic analysis of massive point-vortex dynamics in Bose--Einstein condensates in the small-mass limit $\varepsilon \to 0$. We define two distinguished manifolds in the phase space of the dynamics. We call the first the kinematic subspace $\mathcal{K}$, whereas the second is an almost-invariant set $\mathcal{S}$ called a ``slow manifold.'' The orthogonal projection of the massive dynamics to $\mathcal{K}$ yields the standard massless vortex dynamics or the Kirchhoff equations---also the 0th-order approximation to the massive equation as $\varepsilon \to 0$. Our first main result proves that the massive dynamics starting $O(\varepsilon)$-close to $\mathcal{K}$ remains $O(\varepsilon)$-close to the massless dynamics for short times. The second main result is the derivation of a normal form for the system's Hamiltonian for the two-vortex case; it describes the coupling between motion within $\mathcal{S}$ and that transverse to it. Specifically, we use the Lie transformation perturbation method to derive the first few terms in a formal expansion for $\mathcal{S}$ and demonstrate numerically that fast oscillations due to the vortices' mass are suppressed, given initial conditions sufficiently close to $\mathcal{S}$.
\end{abstract}

\section{Introduction}
The standard description of vortex motion in superfluids has conventionally treated vortices as massless topological defects, leading to the familiar point-vortex (or Kirchhoff) equations~\cite{Lin1941, Lamb1945}.
However, in many real experimental settings, the cores of quantum vortices in superfluids are actually filled, either intentionally or unintentionally, by massive particles, so that the vortices acquire a small but nonzero effective inertial mass.
Examples include tracer atoms introduced in superfluid liquid helium to visualize vortex lines~\cite{Bewley2006, Griffin2020, Peretti2023, Tang2023, Minowa2025}, quasiparticle bound states in both fermionic~\cite{Simanek1995, Kopnin1998, Simula2018, Kwon2021, Barresi2023, Richaud2025Fermi, Grani2025} and bosonic~\cite{Simula2018} superfluids even at zero temperature.
Moreover, experimental conditions inherently involve non-zero temperatures, leading to the presence of thermal atoms localized in vortex cores, both in atomic Bose--Einstein condensates (BECs)~\cite{Griffin2009} and in Fermi superfluids (see Refs.~\cite{Barresi2023, Richaud2025Fermi, Kwon2021, Hernandez2024} and references therein).
Additionally, in two-component BECs, atoms of one species can become localized within the vortex cores of the other species~\cite{Anderson2000, Law2010, PhysRevA.101.013630, PhysRevA.103.023311, Patrick2023, Doran2024, Bellettini2024PRA, Bellettini2024PRR, Kanjo2024,Dambroise2025}.
These examples highlight that neglecting vortex mass can overlook crucial properties of the dynamics, motivating a shift from the purely massless approximation to massive point-vortex models~\cite{PhysRevA.101.013630, PhysRevA.103.023311, Richaud2022rk,Caldara2023, Kanjo2024,Caldara2024, Khalifa2024,Dambroise2025}, which are essential for describing realistic quantum fluids.

In fact, even if the vortex mass is relatively small compared to that of the whole superfluid, it can significantly alter the dynamics of vortices by introducing non-trivial inertial effects.
As a result, one observes oscillatory phenomena~\cite{PhysRevA.103.023311, Richaud2022rk, Kanjo2024, Khalifa2024,Dambroise2025} and collisions~\cite{Richaud2023Collisions} absent within the purely massless descriptions.
These issues are also central in systems involving multicomponent~\cite{Baroni2024, Trabucco2025} or strongly coupled superfluids, where vortex-bound structures can display unusual stability or instability properties~\cite{An2025, An2025_2, Lan2023, Yao2023}.
In fermionic superfluids~\cite{Simanek1995, Kopnin1998, He2023, Magierski2024}, for instance, inertial corrections are deeply connected to the presence of the so-called standard component localized at the vortex core (see the recent Refs.~\cite{Richaud2025Fermi, Levrouw2025} and references therein).
Other examples include vortex–bright-soliton complexes emerging in two-component condensates~\cite{Law2010, Ruban2022, Ruban2022JETP, An2025, Wang2022, Katsimiga2023}.
Moreover, the influence of vortex mass is expected to be relevant in BECs subject to coherent couplings~\cite{Choudhury2022, Choudhury2023}, dipolar interactions~\cite{Prasad2024}, or optical-lattice modulations~\cite{Ancilotto2024}, and for polar-core spin vortices in ferromagnetic spin-1 BECs~\cite{Turner2009, Williamson2021}, underscoring its broad conceptual and experimental importance.

From a mathematical point of view, the introduction of a small core mass introduces \emph{singular perturbation} to the original Kirchhoff equations.
In particular, we examine how the small mass introduced into the system affects the dynamics relative to the massless (Kirchhoff) dynamics, using techniques from singular perturbation theory~\cite{OM1991, KeCo1996, Ne2015}.

In this paper, we perform an asymptotic analysis of the massive point-vortex system in two-component BECs, focusing on how the traditional massless description is recovered at zeroth order and on unveiling the next-order contributions that capture transient inertial features.

Using asymptotic analysis, our results clarify when and how the vortex mass becomes appreciable and how it affects vortex trajectories.
Specifically, we shall identify a subspace $\cK$ in the phase space in which the massive dynamics is approximated well by the massless dynamics; hence $\cK$ is called the \emph{kinematic subspace} here in the sense that the massive dynamics (2nd-order system on the configuration space) approximately behaves like the massless/kinematic system (1st-order system on the configuration space).
Intuitively, this implies that the massive solution starting close to $\cK$ stays close to the corresponding massless solution for a certain amount of time.
Our first main result,~\Cref{thm:massive-massless}, is a rigorous proof of this fact. 

The mass of the vortices enters the equations as a small non-dimensional parameter $\eps$. Solutions near $\cK$ typically display fast oscillations of small amplitude. Our second main result is to compute a normal-form expansion for the Hamiltonian of systems of two massive vortices. This applies when fast-oscillating components are small. To leading order in $\eps$, the dynamics decouples into slow and fast subsystems, the former equivalent to the massless dynamics and the latter describing the fast oscillations. Higher-order terms in the expansion couple the two subsystems. We perform this calculation using the Lie transform perturbation method (LTPM), which preserves the Hamiltonian form of the equations. The LTPM generates a canonical change of variables as a formal series in $\eps$. From this change of variables, we can define, to any order in $\eps$, a nearly invariant manifold $\cS$ called the slow manifold, on which the fast oscillations are suppressed. The kinematic subspace may be considered as a zero-order approximation to $\cS$. Such manifolds are not typically invariant; exponentially small oscillations arise spontaneously in many systems with a slow manifold. Nonetheless, we show numerically that by choosing initial conditions close to $\cS$ we can suppress the fast oscillations on timescales of interest.

\section{Massive Point Vortex Dynamics}
\subsection{Lagrangian for $N$ Massive Vortices in Two-Component BEC}
We consider a two-component Bose--Einstein Condensate in the immiscible regime, confined in a quasi-two-dimensional disk trap of radius $R$, with hard-wall boundary conditions. The majority component, labeled ``$a$'', hosts $N$ quantized vortices of topological charges $\{ q_{j} \in \{\pm1\} \}_{j=1}^{N}$ while the minority component, labeled ``$b$'', occupies the vortex cores. The total masses of the two components are given by $M_a\defeq N_a m_a$ and $M_b \defeq N_b m_b$, where $N_a$ and $N_b$ denote the total number of atoms in each component, and $m_a$ and $m_b$ are their respective atomic masses. In the absence of component $b$, the vortices in component $a$ follow the standard first-order dynamics characteristic of massless point vortices or the Kirchhoff equations~\cite[Ch. 7]{Lamb1945}. However, in the immiscible regime, the minority component $b$ becomes localized within the vortex cores, effectively endowing the vortices with mass. Assuming that the total mass $M_b$ is evenly distributed among the $N$ vortices, each vortex acquires an effective mass of $M_b/N$~\cite{PhysRevA.101.013630, PhysRevA.103.023311}.

Let $\br_{j} = (x_{j}, y_{j})$ be the position of the $j$-th vortex, and $r_{j} \defeq |\br_{j}| = (\br_{j} \cdot \br_{j})^{1/2}$ be its length.
We shall also use $\br \defeq (\br_{1}, \dots, \br_{N})$ for short and similarly $\dot{\br} \defeq (\dot{\br}_{1}, \dots, \dot{\br}_{N})$, $\bq\defeq (q_{1}, \dots, q_{N})$, etc.
According to Richaud et al.~\cite{PhysRevA.101.013630,PhysRevA.103.023311}, the Lagrangian for the $N$ massive vortices reads
\begin{equation}
    \label{eq:L_dimensionful}
    L(\br,\dot{\br}) \defeq \sum_{j=1}^N\left( \frac{M_b}{2 N}\dot{\br}_{j}^{2} +\pi n_a \hbar q_j (\dot{\br}_{j} \times \br_{j}) \cdot \hatbz \right) - E(\br),
\end{equation}
where $n_a$ is the two-dimensional number density, the cross product $\mathbf{a} \times \mathbf{b}$ of two planar vectors $\mathbf{a}, \mathbf{b} \in \R^{2}$ are taken by attaching zero as the third components to both, and is seen as a vector in $\R^{2}$ or $\R^{3}$ depending on the context, $\hatbz \defeq (0,0,1)$, and $E$ is the energy of the standard (massless) vortices
\begin{equation}
    E(\br) \defeq \frac{\pi n_a \hbar^2}{m_a} \sum_{j=1}^N  \ln \left(1-\frac{r_j^2}{R^2}\right) +\frac{\pi n_a \hbar^2 }{m_a}\sum_{1\le j<k\le N}q_j q_k \ln\left(\frac{R^2 -2\br_{j}\cdot\br_{k} +r_j^2r_k^2/R^2 }{r_j^2-2\br_{j}\cdot\br_{k}+\br_{k}^2}\right).
\end{equation}

It is useful to express Lagrangian~\eqref{eq:L_dimensionful} in rescaled non-dimensional variables. We set the disk radius $R$ as the unit of length, $m_aR^2/\hbar$ as the unit of time, and $\pi n_a \hbar^2/m_a$ as the unit of energy. In these natural units, which we adopt henceforth, the Lagrangian takes the form: 
\begin{equation}
  \label{eq:L}
  L(\br,\dot{\br}) \defeq \sum_{j=1}^{N} \parentheses{
    \frac{\eps}{2} \dot{\br}_{j}^{2}
    + q_{j} (\dot{\br}_{j} \times \br_{j}) \cdot \hatbz
  }
  - E(\br),
\end{equation}
where the only effective parameter $\eps$ is defined as
\begin{equation*}
  \eps \defeq \frac{M_{b}/M_{a}}{N},
\end{equation*}
and the potential-energy term reads
\begin{equation}
  \label{eq:E}
  E(\br)
  \defeq \sum_{j=1}^{N} \ln(1 - r_{j}^{2})
  + \sum_{1\le j < k \le N} q_{j} q_{k} \ln\parentheses{ \frac{1 - 2 \br_{j} \cdot \br_{k} + r_{j}^{2} r_{k}^{2}}{|\br_{j} - \br_{k}|^{2}} }.
\end{equation}

\begin{remark}
  Due to the assumed confining potential, strictly speaking, one must restrict the positions of the vortices to the open unit disk, i.e.,
  \begin{equation*}
    \br_{j} \in \mathbb{B}_{1}(0) \defeq \Set{ \mathbf{x} \in \R^{2} |\, |\mathbf{x}| < 1 }
    \quad\text{for}\quad
    1 \le j \le N,
  \end{equation*}
  and also have to avoid collision points---those points with $\br_{i} = \br_{j}$ with $i \neq j$---so that the above energy function $E$ is defined.
  However, in what follows, we shall ignore these issues for simplicity unless otherwise stated, because the geometric consideration to follow is simpler to describe by assuming $\br \in \R^{2N}$.
  Alternatively, one may extend the definition of $E$ to the entire $\R^{2N}$ by assigning values to $E$ outside $( \mathbb{B}_{1}(0) )^{N}$ and collision points.
  Either way, one may consider the dynamics only in the subset of $\R^{2N}$ in which the original $E$ makes sense.
\end{remark}
 
We are mainly interested in the asymptotic behavior of the massive point vortices in the small-mass limit, i.e., $0 < \eps \ll 1$.
This regime is particularly relevant from a physical perspective, as quantum vortices in real superfluid systems are rarely, if ever, truly massless and often have small masses, as explained in the Introduction.

\subsection{Hamiltonian Formulation}
Using the Lagrangian~\eqref{eq:L}, the Legendre transformation is defined via the momenta $\bp \defeq (\bp_{1}, \dots, \bp_{N})$ with
\begin{equation}
  \label{eq:p}
  \bp_{j} \defeq \pd{L}{\dot{\br}_{j}}
  = \eps\,\dot{\br}_{j} + q_{j}(\br_{j} \times \hatbz),
\end{equation}
giving
\begin{equation*}
  \dot{\br}_{j} = \frac{1}{\eps} \parentheses{ \bp_{j} - q_{j}(\br_{j} \times \hatbz) }.
\end{equation*}
Hence, we may define the Hamiltonian
\begin{align}
  \label{eq:H}
  H(\br,\bp)
  &\defeq \sum_{j=1}^{N} \bp_{j} \cdot \dot{\br}_{j} - L(\br,\dot{\br}) \nonumber\\
  &= \frac{1}{2\eps} \sum_{j=1}^{N} \parentheses{ \bp_{j} - q_{j}(\br_{j} \times \hatbz) }^{2} + E(\br) \nonumber\\
  &= \frac{1}{2\eps} \sum_{j=1}^{N} \parentheses{ p_{j}^{2} - 2q_{j}(\br_{j} \times \hatbz) \cdot \bp_{j} + q_{j}^{2} r_{j}^{2} } + E(\br).
\end{align}
Then Hamilton's equations
\begin{equation*}
  \dot{\br}_{j} = \pd{H}{\bp_{j}},
  \qquad
  \dot{\bp}_{j} = -\pd{H}{\br_{j}}
\end{equation*}
yield
\begin{equation*}
  \eps\,\dot{\br}_{j} = \bp_{j} - q_{j}(\br_{j} \times \hatbz),
  \qquad
  \eps\,\dot{\bp}_{j} = q_{j}(\hatbz \times \bp_{j}) - q_{j}^{2} \br_{j} - \eps\,\nabla_{j}E(\br),
\end{equation*}
where $\nabla_{j}$ stands for $\tpd{}{\br_{j}}$, i.e., the gradient with respect to $\br_{j}$.

For the sake of brevity, let us define
\begin{equation}
  \label{eq:J}
  J \defeq
  \begin{bmatrix}
    0 & 1 \\
    -1 & 0
  \end{bmatrix}
  \text{ so that }
  J \mathbf{a} = \mathbf{a} \times \hatbz
  \quad
  \forall \mathbf{a} \in \R^{2}.
\end{equation}
Then we have
\begin{equation}
  \label{eq:Hamilton}
  \eps\,\dot{\br}_{j} = \bp_{j} - q_{j} J \br_{j},
  \qquad
  \eps\,\dot{\bp}_{j} = -q_{j} J (\bp_{j} - q_{j} J \br_{j}) - \eps\,\nabla_{j}E(\br).
\end{equation}
One may eliminate $\bp_{j}$ from above to obtain the following second-order differential equations
\begin{equation}
  \label{eq:Euler-Lagrange}
  \eps\,\ddot{\br}_{j} + 2 q_{j} J \dot{\br}_{j} = - \nabla_{j}E(\br),
\end{equation}
which are nothing but the Euler--Lagrange equations for the Lagrangian $L$ in~\eqref{eq:L}.
Clearly, these systems are singularly perturbed systems when $\eps \ll 1$.

The Hamiltonian formulation is beneficial in highlighting the drastic effect of introducing a core mass.
In the massless case, in fact, $x_j$ is canonically conjugate to $q_j\,y_j$, i.e., the massless $N$-vortex system is a Hamiltonian system with $N$ degrees of freedom (associated to a $2N$-dimensional phase space).
However, in the massive case, the system gains $2N$ independent momenta, doubling the dimension of the associated phase space to $4N$.
Despite this increase, the number of conserved quantities, typically arising from the symmetry properties of the superfluid domain, remains unchanged.
As a result, the introduction of a core mass can break integrability.
A clear example is the two-vortex system in a disk.
In the massless limit, there are only two degrees of freedom, matching the number of conserved quantities (energy and angular momentum) and therefore ensuring integrability~\cite{Boffetta1996}.
In contrast, the massive case has four degrees of freedom while the number of conserved quantities remains two.
Hence, the system may be chaotic unless hidden invariants exist.
This distinction is particularly relevant for long-time dynamics, where slight differences in initial conditions, such as identical positions but slightly different velocities, may be exponentially amplified, leading to drastically different final states.

\section{Geometry of the Kinematic Approximation}
\subsection{The Kinematic Approximation and the Kirchhoff Equation}
Taking the limit $\eps \to 0$ in~\eqref{eq:Euler-Lagrange}, one obtains
\begin{equation}
  \label{eq:Kirchhoff}
  2q_{j} J \dot{\br}_{j} = -\nabla_{j}E(\br)
  \iff
  \left\{
    \begin{array}{l}
      \DS 2q_{j} \dot{x}_{j} = \pd{E}{y_{j}}, \medskip\\
      \DS 2q_{j} \dot{y}_{j} = -\pd{E}{x_{j}}.
    \end{array}
  \right.
\end{equation}
These are nothing but the Kirchhoff equations for massless vortices, and give the natural kinematic or massless approximation to the massive point vortex equations~\eqref{eq:Euler-Lagrange}.

According to these equations, each vortex moves with the local superfluid velocity at its position, determined by the superposition of all velocity fields induced by all the other vortices in the system. To be more specific, the velocity field generated at a point $\br$ by the $j$-th vortex, located at $\br_j$ and having topological charge $q_j$, is $ \bv(\br) = q_j \hatbz \times (\br - \br_j)/|\br - \br_j|^2$.
Hence, each vortex moves with the velocity obtained by summing the contributions from all other vortices (except itself), be they physical vortices or image vortices ensuing from the presence of boundaries. This framework is well established in classical hydrodynamics, where Euler’s equations describe the evolution of an inviscid and incompressible fluid~\cite{Lamb1945, Boffetta1996}.

\subsection{Kinematic Subspace $\cK$}
Here, we aim to give a geometric interpretation of the above kinematic approximation in the limit $\eps \to 0$.
First notice that taking the limit $\eps \to 0$ in~\eqref{eq:Hamilton} gives
\begin{equation}
  \label{eq:kinematic_constraint}
  \bp_{j} = q_{j} J \br_{j}
  \iff
  \begin{bmatrix}
    \xi_{j} \\
    \eta_{j}
  \end{bmatrix}
  = q_{j}
  \begin{bmatrix}
    y_{j} \\
    -x_{j}
  \end{bmatrix}
\end{equation}
where we wrote $\bp_{j} = (\xi_{j}, \eta_{j})$.
One may interpret the above constraints as the natural kinematic constraints in the massless limit $\eps \to 0$ in the definition~\eqref{eq:p} of the momentum $\bp$ (or the Legendre transformation).

Defining the cotangent bundle---the phase space for the Hamiltonian system~\eqref{eq:Hamilton}---as
\begin{equation*}
  T^{*}\R^{2N} = \Set{ (\br,\bp) \in \R^{2N} \times \R^{2N} | \bp \in T_{\br}^{*}\R^{2N} \cong \R^{2N} } \cong \R^{4N},
\end{equation*}
the above kinematic constraints~\eqref{eq:kinematic_constraint} give rise to the following \textit{kinematic subspace}:
\begin{equation*}
  \cK \defeq \Set{ (\br,\bp) \in T^{*}\R^{2N} | \bp_{j} = q_{j} J \br_{j} \text{ for } 1\le j \le N  }.
\end{equation*}
Clearly, $\cK$ is an $2N$-dimensional subspace of $T^{*}\R^{2N} \cong \R^{4N}$.
One also sees that $\br = (\br_{1}, \dots \br_{N}) \in \R^{2N}$ gives coordinates for $\cK$ via the following map (see~\Cref{fig:kinematic_subspace}):
\begin{equation}
  \label{eq:varphi}
  \varphi\colon \R^{2N} \to \cK;
  \qquad
  \br = (\br_{1}, \dots \br_{N})
  \mapsto
  \varphi(\br) \defeq (\br_{1}, \dots, \br_{N}, q_{1} J \br_{1}, \dots, q_{N} J \br_{N}),
\end{equation}
and hence also giving the identification
\begin{equation*}
  \cK = \varphi(\R^{2N}) \cong \R^{2N}.
\end{equation*}

\begin{figure}[htbp]
  \centering
  \includegraphics[width=.7\linewidth]{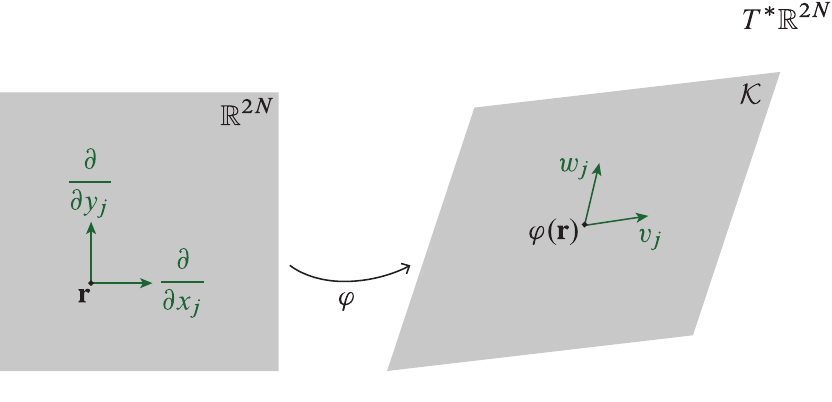}
  \caption{
    Subspace $\cK$ and tangent vectors $v_{j}, w_{j} \in T_{\varphi(\br)}\cK$ forming a basis $\{ v_{j}, w_{j} \}_{j=1}^{N}$ for $T_{\varphi(\br)}\cK$.
    Note that, strictly speaking, the tangent vectors in the figure live in tangent spaces of $\R^{2N}$ and $\cK$ but are drawn in the base spaces for simplicity.
  }
  \label{fig:kinematic_subspace}
\end{figure}

The Jacobian of $\varphi$ is given by
\begin{equation*}
  D\varphi(\br) =
  \begin{bmatrix}
    I & 0 \\
    0 & I \\
    0 & \mathsf{D}_{q} \\
    -\mathsf{D}_{q} & 0
  \end{bmatrix}
  \quad
  \text{with}
  \quad
  \mathsf{D}_{q} \defeq \diag(q_{1}, \dots, q_{N}).
\end{equation*}
Thus the tangent vectors $\pd{}{x_{j}}$ and $\pd{}{y_{j}}$ in $T_{\br}\R^{2N}$ are mapped to the following vectors in $T_{\varphi(\br)}\cK$ by $D\varphi$ as follows:
\begin{equation}
  \label{eq:v_j-w_j}
  v_{j} \defeq D\varphi(\br) \pd{}{x_{j}} = \pd{}{x_{j}} - q_{j} \pd{}{\eta_{j}},
  \qquad
  w_{j} \defeq D\varphi(\br) \pd{}{y_{j}} = \pd{}{y_{j}} + q_{j} \pd{}{\xi_{j}},
\end{equation}
which are written as tangent vectors in $T_{\varphi(\br)}T^{*}\R^{2N} = \Span\{ \pd{}{x_{j}}, \pd{}{y_{j}}, \pd{}{\xi_{j}}, \pd{}{\eta_{j}} \}_{j=1}^{N}$; see~\Cref{fig:kinematic_subspace}.
As a result, we have the following basis for the tangent space $T_{\varphi(\br)}\cK$:
\begin{equation}
  \label{eq:basis-TmathcalK}
  \braces{
    v_{j} \defeq \pd{}{x_{j}} - q_{j} \pd{}{\eta_{j}},
    \,
    w_{j} \defeq \pd{}{y_{j}} + q_{j} \pd{}{\xi_{j}}
  }_{j=1}^{N}.
\end{equation}

\subsection{Geometric Interpretation of the Kinematic Approximation}
We shall show that the Kirchhoff equation~\eqref{eq:Kirchhoff} is a natural Hamiltonian system induced in the kinematic subspace $\cK \subset T^{*}\R^{2N}$ by the Hamiltonian system~\eqref{eq:Hamilton} for massive point vortices.
To see this, first note that the standard symplectic form on $T^{*}\R^{2N}$ is given by
\begin{equation}
  \label{eq:Omega}
  \Omega \defeq \d\br_{j} \wedge \d\bp_{j} = \d{x}_{j} \wedge \d{\xi_{j}} + \d{y}_{j} \wedge \d{\eta_{j}},
\end{equation}
where $\d$ stands for the exterior derivative, and the summation convention is assumed on $j$; its matrix representation is the $4N \times 4N$ matrix
\begin{equation*}
  \mathbb{J} \defeq
  \begin{bmatrix}
    0 & I \\
    -I & 0
  \end{bmatrix}
\end{equation*}
with $I$ being the $2N \times 2N$ identity matrix; so we have $\Omega(v,w) = v^{T} \mathbb{J} w$ with all $v, w \in T_{(\br,\bp)} T^{*}\R^{2N} \cong \R^{4N}$.

Writing vector field $X_{H}$ on $T^{*}\R^{2N}$ as
\begin{equation*}
  X_{H} = \dot{\br}_{j} \pd{}{\br_{j}} + \dot{\bp}_{j} \pd{}{\bp_{j}}
  = \dot{x}_{j} \pd{}{x_{j}} + \dot{y}_{j} \pd{}{y_{j}} + \dot{\xi}_{j} \pd{}{\xi_{j}} + \dot{\eta}_{j} \pd{}{\eta_{j}},
\end{equation*}
the Hamiltonian system~\eqref{eq:Hamilton} is then equivalent to
\begin{equation*}
  \ins{X_{H}}{\Omega} = \d{H},
\end{equation*}
where $\ins{}{}$ stands for the interior product (or insertion/contraction) of a vector field with a differential form.
More concretely, one may think of $X_{H}$ and $\nabla H$ as column vectors in $\R^{4N}$ so that the above system becomes equivalent to
\begin{equation*}
  X_{H}^{T} \mathbb{J} = (\nabla H)^{T}
  \iff
  X_{H} = \mathbb{J} \nabla H.
\end{equation*}
As a result, we have
\begin{equation}
  \label{eq:X_H}
  X_{H} = \frac{1}{\eps}(\bp_{j} - q_{j} J \br_{j}) \pd{}{\br_{j}}
  - \parentheses{ \frac{1}{\eps} q_{j} J (\bp_{j} - q_{j} J \br_{j}) + \nabla_{j}E(\br) } \pd{}{\bp_{j}}
\end{equation}

With an abuse of notation, we may change the codomain of $\varphi$ to $T^{*}\R^{2N}$ and consider
\begin{equation*}
  \varphi\colon \R^{2N} \to T^{*}\R^{2N}.
\end{equation*}
Then $\varphi$ gives an embedding of $\R^{2N}$ into $T^{*}\R^{2N}$ as the subspace $\cK = \varphi(\R^{2N})$.
Let us consider the pull-back of the symplectic form $\Omega$ to $\cK \cong \R^{2N}$ by $\varphi$:
\begin{equation*}
  \omega \defeq \varphi^{*}\Omega = 2 \sum_{j=1}^{N} q_{j}\,\d{x}_{j} \wedge \d{y_{j}},
\end{equation*}
where $\varphi^{*}$ stands for the pull-back by $\varphi$, and we did \textit{not} use the summation convention here; its matrix representation is
\begin{equation*}
  \mathbb{K} \defeq (D\varphi(\br))^{T} \mathbb{J}\, D\varphi(\br)
  = 2
  \begin{bmatrix}
    0 & \mathsf{D}_{q} \\
    -\mathsf{D}_{q} & 0
  \end{bmatrix}.
\end{equation*}
Similarly, the pull-back of the Hamiltonian $H$ to $\cK \cong \R^{2N}$ by $\varphi$ gives
\begin{equation*}
  (\varphi^{*}H)(\br) = H \circ \varphi(\br) = E(\br),
\end{equation*}
where the last equality follows by using the expression~\eqref{eq:H} and the kinematic constraints~\eqref{eq:kinematic_constraint} defining $\cK$.

Writing vector field $X_{E}$ on $\cK \cong \R^{2N}$ as
\begin{equation*}
  X_{E} = \dot{\br}_{j} \pd{}{\br_{j}}
  = \dot{x}_{j} \pd{}{x_{j}} + \dot{y}_{j} \pd{}{y_{j}},
\end{equation*}
we have a Hamiltonian system defined on $\R^{2N}$ as
\begin{equation*}
  \ins{X_{E}}{\omega} = \d{E}
\end{equation*}
in terms of the symplectic form $\omega$ and Hamiltonian $E$ naturally induced on $\R^{2N}$ as the pull-backs by $\varphi$ of $\Omega$ and $H$; see~\Cref{fig:projection}.
Seeing $X_{E}$ and $\nabla E$ as column vectors in $\R^{2N}$, the above system is equivalent to
\begin{equation*}
  X_{E}^{T}\, \mathbb{K} = (\nabla{E})^{T}
  \iff
  X_{E} = (\mathbb{K}^{T})^{-1} \nabla{E}
  = \frac{1}{2}
  \begin{bmatrix}
    0 & \mathsf{D}_{q} \\
    -\mathsf{D}_{q} & 0
  \end{bmatrix}
  \nabla E,
\end{equation*}
noting that $q_{j}^{-1} = q_{j}$ because $q_{j} = \pm1$; as a result, we have
\begin{equation}
  \label{eq:X_E}
  X_{E}(\br) = \frac{q_{j}}{2} \pd{E}{y_{j}} \pd{}{x_{j}} - \frac{q_{j}}{2} \pd{E}{x_{j}} \pd{}{y_{j}},
\end{equation}
which is the vector field defined by the Kirchhoff equation~\eqref{eq:Kirchhoff}.

\begin{figure}[htbp]
  \centering
  \includegraphics[width=.7\linewidth]{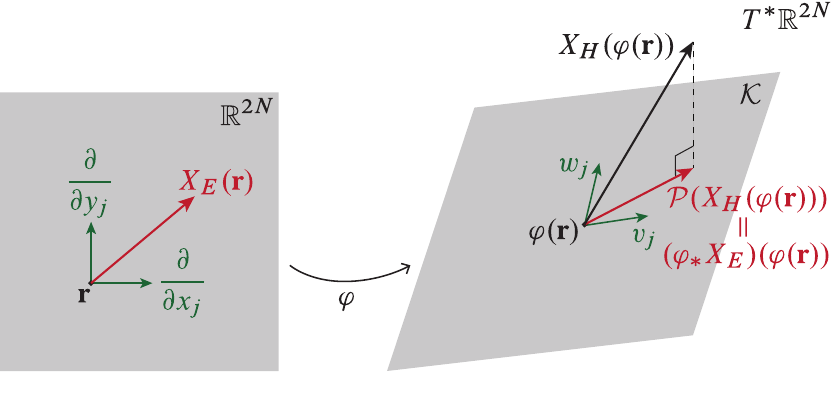}
  \caption{
    Orthogonally projecting $X_{H}$---the vector field defining the massive point vortex equation~\eqref{eq:Hamilton}---from the tangent space of $T^{*}\R^{2N}$ to the tangent space of $\cK$ defines vector field $\mathcal{P}(X_{H})$ on $\cK$; this in turn corresponds via $\varphi$ to vector field $X_{E}$ on $\R^{2N}$ defining the Kirchhoff equation~\eqref{eq:Kirchhoff}.
  }
  \label{fig:projection}
\end{figure}

\subsection{Best Approximation Property}
We shall show that the vector field $X_{E}$ in $\R^{2N}$---more precisely its push-forward $\varphi_{*}X_{E}$ on $\cK$---is the best approximation to $X_{H}$ in $\cK$ in the following sense (see also~\Cref{fig:projection}):
\begin{proposition}
  Let $g$ be the standard metric on $T^{*}\R^{2N}$, i.e.,
  \begin{equation}
    \label{eq:g}
    g = \d{x_{j}} \otimes \d{x_{j}} + \d{y_{j}} \otimes \d{y_{j}} + \d{\xi_{j}} \otimes \d{\xi_{j}} + \d{\eta_{j}} \otimes \d{\eta_{j}},
  \end{equation}
  and consider the Hamiltonian vector field $X_{H}$ on $T^{*}\R^{2N}$ (from~\eqref{eq:X_H}) for the massive vortex dynamics~\eqref{eq:Hamilton} as well as the Hamiltonian vector field $X_{E}$ on $\R^{2N}$ (from~\eqref{eq:X_E}) for the massless vortex dynamics~\eqref{eq:Kirchhoff}.
  Then the orthogonal projection of $X_{H}$ to the kinematic subspace $\cK$ is the push-forward $\varphi_{*}X_{E}$ of $X_{E}$ by $\varphi\colon \R^{2N} \to \cK$ defined in~\eqref{eq:varphi}.
\end{proposition}
\begin{proof}
  First note that the basis~\eqref{eq:basis-TmathcalK} for $T_{\varphi(\br)}\cK$ is an orthogonal basis with respect to the metric $g$:
  \begin{gather*}
    g(v_{i},v_{j}) = \delta_{ij} \norm{v_{j}}^{2},
    \qquad
    g(w_{i},w_{j}) = \delta_{ij} \norm{w_{j}}^{2},
    \qquad
    g(v_{i},w_{j}) = 0
  \end{gather*}
  with
  \begin{equation*}
    \norm{v_{j}} \defeq g(v_{j},v_{j})^{1/2} = (1 + q_{j}^{2})^{1/2} = \sqrt{2},
    \qquad
    \norm{w_{j}} \defeq g(w_{j},w_{j})^{1/2} = (1 + q_{j}^{2})^{1/2} = \sqrt{2},
  \end{equation*}
  because we assume that $q_{j} = \pm1$.
  We may then define the orthogonal projection
  \begin{equation*}
    \mathcal{P} \colon T_{\varphi(\br)}T^{*}\R^{2N} \to T_{\varphi(\br)}\cK;
    \qquad
    V \mapsto \frac{g(v_{j}, V)}{\norm{v_{j}}^{2}} v_{j}
    + \frac{g(w_{j}, V)}{\norm{w_{j}}^{2}} w_{j}
  \end{equation*}
  using the summation convention on $j$.
  
  Notice that $X_{H}(\varphi(\br))$, upon imposing the kinematic constraints~\eqref{eq:kinematic_constraint} to the expression~\eqref{eq:X_H}, takes the form
  \begin{equation*}
    X_{H}(\varphi(\br)) = -\pd{E}{x_{j}} \pd{}{\xi_{j}} -\pd{E}{y_{j}} \pd{}{\eta_{j}}
  \end{equation*}
  and thus, using the expressions for $v_{j}$ and $w_{j}$ from~\eqref{eq:basis-TmathcalK},
  \begin{equation*}
    g(v_{j},X_{H}(\varphi(\br)) ) = q_{j} \pd{E}{y_{j}},
    \qquad
    g(w_{j},X_{H}(\varphi(\br)) ) = -q_{j} \pd{E}{x_{j}} 
  \end{equation*}
  Hence, the orthogonal projection of $X_{H}(\varphi(\br))$ to $T_{\varphi(\br)}\cK$ is as follows:
  \begin{align*}
    \mathcal{P}(X_{H}(\varphi(\br)))
    &= \frac{g(v_{j}, X_{H}(\varphi(\br)) )}{\norm{v_{j}}^{2}} v_{j}
      + \frac{g(w_{j}, X_{H}(\varphi(\br)) )}{\norm{w_{j}}^{2}} w_{j} \\
    &= \frac{q_{j}}{2} \pd{E}{y_{j}} v_{j}
      - \frac{q_{j}}{2} \pd{E}{x_{j}} w_{j}.
  \end{align*}
  Sending this vector field from $\cK$ to $\R^{2N}$ using the correspondence $v_{j} \leftrightarrow \pd{}{x_{j}}$ and $w_{j} \leftrightarrow \pd{}{y}_{j}$ from~\eqref{eq:v_j-w_j} (i.e., the pull-back by $\varphi$), we obtain the following vector field on $\R^{2N}$:
  \begin{equation*}
    \frac{q_{j}}{2} \pd{E}{y_{j}} \pd{}{x_{j}} - \frac{q_{j}}{2} \pd{E}{x_{j}} \pd{}{y_{j}},
  \end{equation*}
  which is $X_{E}(\br)$, defining the Kirchhoff equation~\eqref{eq:Kirchhoff}:
  \begin{equation*}
    \dot{x}_{j} = \frac{q_{j}}{2} \pd{E}{y_{j}},
    \qquad
    \dot{y}_{j} = -\frac{q_{j}}{2} \pd{E}{x_{j}},
  \end{equation*}
  where we again note that $q_{j}^{-1} = q_{j}$ because $q_{j} = \pm1$.
\end{proof}

We note that the above proposition is reminiscent of the so-called variational approximation or the Dirac--Frenkel--McLachlan variational principle~\cite{Di1930, Fr1934, Mc1964, KrSa1981, Lu2005, Lu2008} that is often used in approximations of quantum dynamics.

\subsection{Numerical Results: Massless vs.~Massive near $\mathcal{K}$}
\label{ssec:numerical_result_on_K}
Consider the vortex dipole case with the following parameters and initial conditions:
\begin{equation}
  \label{eq:IC_on_K}
  \begin{array}{c}
    N = 2,\quad  q_{1} = -1,\quad q_{2} = 1,\quad \eps = 0.01, \medskip\\
    \br_{1}(0) = (x_{1}(0), y_{1}(0)) = (0.6, 0.2),\quad \br_{2}(0) = (x_{2}(0), y_{2}(0)) = (-0.3, -0.4), \medskip\\
    \bp_{1}(0) = q_{1} J \br_{1}(0) = (-0.2,0.6),\quad \bp_{2}(0) = q_{2} J \br_{2}(0) = (-0.4,0.3).
  \end{array}
\end{equation}
Notice that $\bp_{j}(0) = q_{j} J \br_{j}(0)$ so that $(\br(0),\bp(0))$ is in the kinematic subspace $\mathcal{K}$.

\Cref{fig:massless_on_K} compares the massless solution of the Kirchhoff equation~\eqref{eq:Kirchhoff} with the massive solution of~\eqref{eq:Hamilton}.
We observe that the massless and massive solutions are very close to each other, though they slowly diverge.

\begin{figure}[htbp]
  \centering
  \begin{subcaptionblock}[b]{0.475\textwidth}
    \centering
    \includegraphics[width=\textwidth]{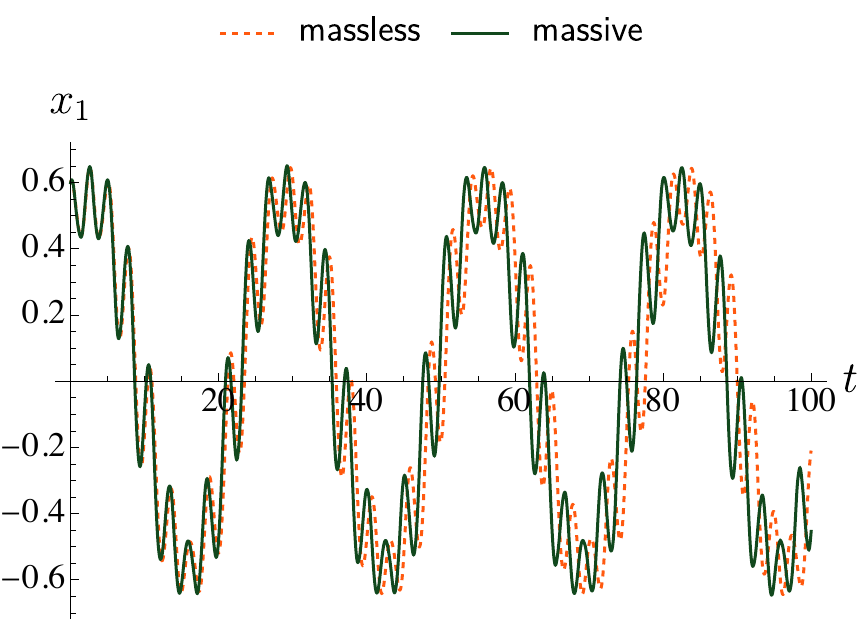}
    \caption{$x$-coordinate of vortex~1}
    \label{fig:t-x01_on_K}
  \end{subcaptionblock}
  \hfill
  \begin{subcaptionblock}[b]{0.5\textwidth}
    \centering
    \includegraphics[width=\textwidth]{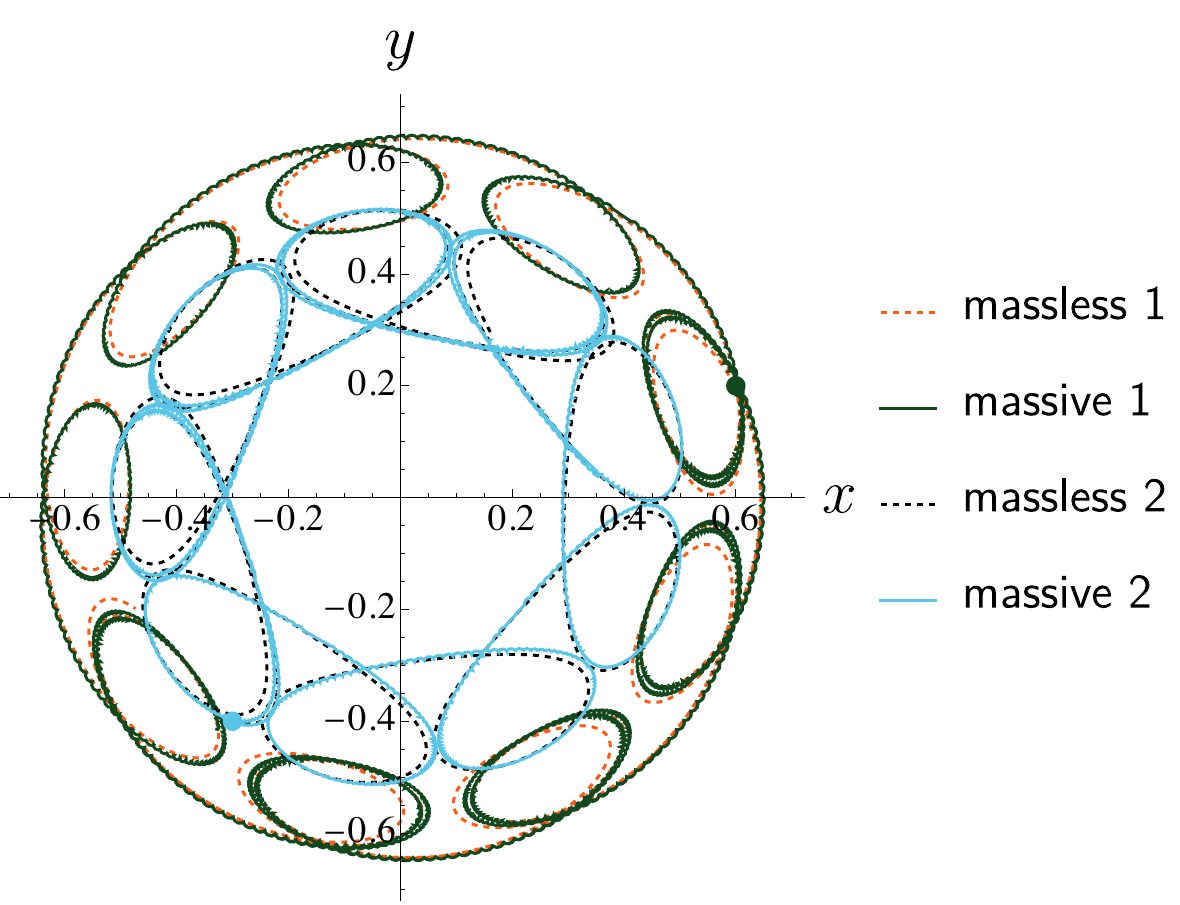}
    \caption{trajectories of vortices on $(x,y)$-plane}
  \end{subcaptionblock}
  \caption{Massless and massive numerical time series of vortex~1 with parameters and initial conditions from~\eqref{eq:IC_on_K} satisfying $(\br(0),\bp(0)) \in \mathcal{K}$.}
  \label{fig:massless_on_K}
\end{figure}

\subsection{Noether Invariants}
The Hamiltonian~\eqref{eq:H} has the $\SO(2)$-symmetry:
For every planar rotation matrix $\mathsf{R} \in \SO(2)$,
\begin{equation*}
  H\parentheses{ \mathsf{R}\br_{1}, \dots, \mathsf{R}\br_{N}, \mathsf{R}\bp_{1}, \dots, \mathsf{R}\bp_{N} } = H(\br_{1}, \dots, \br_{N}, \bp_{1}, \dots, \bp_{N}).
\end{equation*}
The corresponding Noether invariant is the angular momentum
\begin{equation*}
  \ell(\br,\bp) \defeq \sum_{j=1}^{N} (\br_{j} \times \bp_{j}) \cdot \hatbz
  = \sum_{j=1}^{N} (\br_{j} \times \hatbz) \cdot \bp_{j}
  = \sum_{j=1}^{N} (J\br_{j}) \cdot \bp_{j}
\end{equation*}
Its pull-back by $\varphi$ (see~\eqref{eq:varphi}) gives
\begin{equation}
  \label{eq:AngImpulse}
  \mathcal{I}(\br) \defeq (\varphi^{*}\ell)(\br) = \ell \circ \varphi(\br)
  = \sum_{j=1}^{N} (J\br_{j}) \cdot (q_{j} J \br_{j})
  = \sum_{j=1}^{N} q_{j}\, r_{j}^{2},
\end{equation}
and this is the Noether invariant---the so-called angular impulse---of the Kirchhoff equation~\eqref{eq:Kirchhoff} corresponding to the $\SO(2)$-symmetry.

Quantity~\eqref{eq:AngImpulse} is closely related to the third component of the expectation value of the angular momentum of the $N$-vortex superfluid system, given in dimensionful units by $L_{z,a}=\pi \hbar n_a \sum_{j=1}^{N}(R^2-q_jr_j^2)$, and in dimensionless units by $L_{z,a}=\sum_{j=1}^N(1-q_j r_j^2)$.
The angular impulse $\mathcal{I}$ from above differs from this expression only by inessential additive and multiplicative constants.
This formula can be derived by computing the expectation value $L_{z,a}=\langle\psi_a|\hat{L}_z|\psi_a \rangle$ of the angular-momentum operator $\hat{L}_z$ with respect to a suitably chosen trial wavefunction $\psi_a$ capturing the main features of the state of component-$a$ BEC or of another quantum fluid (see Ref.~\cite{Caldara2023} for further details).
We note, in this regard, that the Lagrangian~\eqref{eq:L} is defined up to a total time derivative, allowing for different choices that lead to slight variations in the form of invariant~\eqref{eq:AngImpulse}.
However, among these possible variations, differing only by additive and/or multiplicative constants, only one corresponds directly to the physically meaningful angular momentum of the system. 

\begin{figure}[htbp]
  \centering
  \includegraphics[width=.55\textwidth]{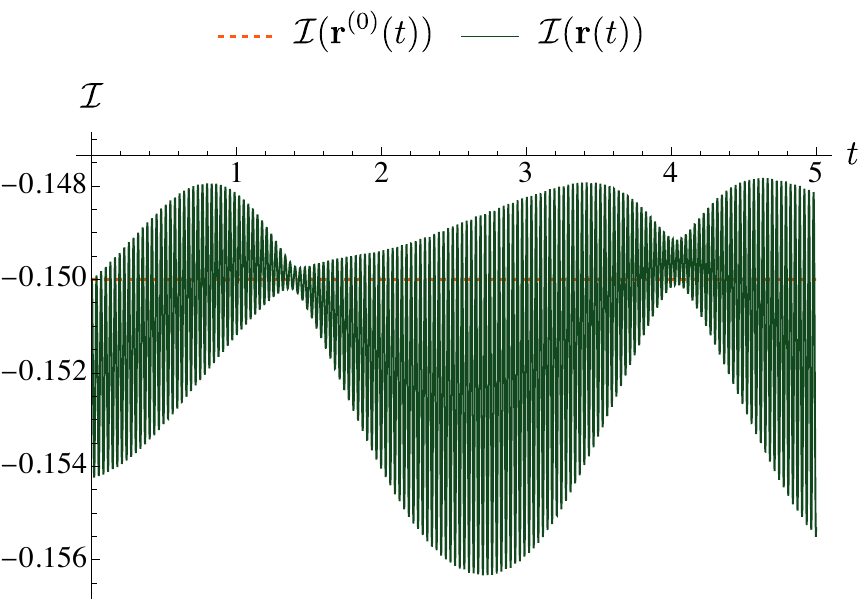}
  \caption{
    Time evolution of  angular impulse $\mathcal{I}$ (see~\eqref{eq:AngImpulse}) along  massless solution $\br^{(0)}(t)$ and massive solution $\br(t)$ with parameters and initial conditions from~\eqref{eq:IC_on_K} satisfying $(\br(0),\bp(0)) \in \mathcal{K}$.
  }
  \label{fig:AngImpulse_on_K}
\end{figure}

\Cref{fig:AngImpulse_on_K} shows the time evolution of the angular impulse $\mathcal{I}$ from~\eqref{eq:AngImpulse} along both the massless solution $\br^{(0)}(t)$ and the massive solution $\br(t)$ with parameters and initial conditions from~\eqref{eq:IC_on_K} satisfying $(\br(0),\bp(0)) \in \mathcal{K}$.
Although $\mathcal{I}$ is \textit{not} an invariant of the massive equations~\eqref{eq:Hamilton}, we see that $\mathcal{I}(\br(t))$ is nearly conserved, oscillating near the actual conserved quantity~$\mathcal{I}(\br^{(0)}(t))$.
This result suggests that the invariant $\mathcal{I}$ of the massless system nearly persists for the massive system when the massive solution stays close to $\mathcal{K}$ when $\eps \ll 1$. Observe that this figure is plotted over a short time interval compared with~\cref{fig:massless_on_K}; the impulse varies on two time scales, with oscillations of similar amplitude on both.

\section{Averaging and Approximation by Massless Solutions}
In this section, we relate the massive system to the massless system via the method of averaging.
Specifically, we consider the periodic averaging of the massive system~\eqref{eq:Hamilton} and show, in~\Cref{prop:averaging} and~\Cref{cor:averaging}, that the averaged system is \textit{approximately} the massless system~\eqref{eq:Kirchhoff}.
This leads to our first main result,~\Cref{thm:massive-massless}, that the massive system starting $O(\eps)$-close to $\mathcal{K}$ is approximated by the corresponding massless system with $O(\eps)$ error.

\subsection{Massive System in Fast Time}
\label{ssec:massive-fast_time}
To facilitate the averaging procedure to be discussed in this section, let us first rewrite the massive system~\eqref{eq:Hamilton} as
\begin{equation*}
  \eps\,\od{\bz_{j}}{t} = A_{j} \bz_{j} + \eps\,\mathbf{f}_{j}(\bz)
  \quad
  \text{for}
  \quad
  1 \le j \le N,
\end{equation*}
where we defined
\begin{equation}
  \label{eq:z_A_f}
  \bz_{j} \defeq
  \begin{bmatrix}
    \br_{j} \\
    \bp_{j}
  \end{bmatrix},
  \qquad
  A_{j} \defeq
  \begin{bmatrix}
    -q_{j} J & I \\
    -I & -q_{j} J
  \end{bmatrix},
  \text{ and }
  \mathbf{f}_{j}(\bz) \defeq
  \begin{bmatrix}
    0 \\
    -\nabla_{j}E(\br)
  \end{bmatrix}.
\end{equation}
Furthermore, setting
\begin{equation*}
  \bz \defeq
  \begin{bmatrix}
    \bz_{1} \\
    \vdots \\
    \bz_{N}
  \end{bmatrix},
  \qquad
  A \defeq \begin{bmatrix}
    A_{1} &  & 0\\
          & \ddots & \\
    0 & & A_{N}
  \end{bmatrix},
  \text{ and }
  \mathbf{f}(\bz) \defeq
  \begin{bmatrix}
    \mathbf{f}_{1}(\bz) \\
    \vdots \\
    \mathbf{f}_{N}(\bz)
  \end{bmatrix},
\end{equation*}
we have
\begin{equation}
  \label{eq:massive_vec}
  \eps\,\od{}{t}\bz(t) = A\, \bz(t) + \eps\,\mathbf{f}(\bz(t)).
\end{equation}

Using the fast time
\begin{equation*}
  T \defeq \frac{t}{\eps}, 
\end{equation*}
we have
\begin{equation*}
  \od{}{T}\bz(\eps T) = A\, \bz(\eps T) + \eps\,\mathbf{f}(\bz(\eps T))
  \iff
  \od{}{T}\parentheses{ e^{-T A}\, \bz(\eps T) } = \eps\, e^{-T A}\, \mathbf{f}(\bz(\eps T)).
\end{equation*}
Setting
\begin{equation}
  \label{eq:bw}
  \bw(T) =
  \begin{bmatrix}
    \bw_{1}(T) \\
    \vdots \\
    \bw_{N}(T)
  \end{bmatrix}
  \defeq e^{-T A}\, \bz(\eps T)
  =
  \begin{bmatrix}
    e^{-T A_{1}} \mathbf{z}_{1}(\eps T) \\
    \vdots \\
    e^{-T A_{N}} \mathbf{z}_{N}(\eps T)
  \end{bmatrix},
\end{equation}
and
\begin{equation}
  \label{eq:F}
  \mathbf{F}(\bw,T) = 
  \begin{bmatrix}
    \mathbf{F}_{1}(\bw,T) \\
    \vdots \\
    \mathbf{F}_{N}(\bw,T)
  \end{bmatrix}
  \defeq e^{-T A}\, \mathbf{f}\parentheses{ e^{T A}\,\bw }
  = \begin{bmatrix}
    e^{-T A_{1}} \mathbf{f}_{1}\parentheses{ \mathbf e^{T A}\,\bw } \\
    \vdots \\
    e^{-T A_{N}} \mathbf{f}_{N}\parentheses{ \mathbf e^{T A}\,\bw }
  \end{bmatrix},
\end{equation}
we have the following massive system in the fast time $T$:
\begin{equation}
  \label{eq:massive-fast_time}
  \od{\bw}{T} = \eps\, \mathbf{F}(\bw(T),T).
\end{equation}
Such changes of variables are called \emph{van der Pol} transformations.

\subsection{Averaging and Massless System}
Notice that $e^{-2\pi A_{j}} = I$ for every $j \in \{1, \dots, N\}$, and thus $T \mapsto \mathbf{F}(\bw,T)$ is $2\pi$-periodic:
\begin{equation*}
  \mathbf{F}(\bw, T + 2\pi) = \mathbf{F}(\bw, T).
\end{equation*}
The standard periodic averaging theory (see, e.g., \cite[Chapter~2]{SaVeMu2007}) applied to this system tells us the following:
Define the $2\pi$-average of $\mathbf{F}$
\begin{equation}
  \label{eq:overline_F}
  \overline{\mathbf{F}}(\mathbf{x})
  \defeq \frac{1}{2\pi} \int_{0}^{2\pi} \mathbf{F}(\mathbf{x}, \tau)\, d\tau
  = \frac{1}{2\pi} \int_{0}^{2\pi} e^{-\tau A}\, \mathbf{f}\parentheses{ e^{\tau A}\mathbf{x} } d\tau,
\end{equation}
and consider the averaged system
\begin{equation*}
  \od{\mathbf{x}}{T} = \eps\, \overline{\mathbf{F}}(\mathbf{x}).
\end{equation*}
Then, under certain conditions, one has $\norm{ \bw(T) - \mathbf{x}(T) } = O(\eps)$; see, e.g., \cite[Theorem~2.8.1]{SaVeMu2007}.

One might expect that the above averaged system is essentially the massless system.
However, it turns out that the massless vector field is given by the following slightly different averaged system:
\begin{equation}
  \label{eq:G}
  \mathbf{G}(\bw) \defeq \frac{1}{2\pi} \int_{0}^{2\pi} e^{-\tau A}\, \mathbf{f}(\bw)\,d\tau.
\end{equation}
To see this, let $j \in \{1, \dots, N\}$ be arbitrary and first observe that simple calculations yield
\begin{equation*}
  M_{j} \defeq
  \frac{1}{2\pi} \int_{0}^{2\pi} e^{-A_{j} \tau} d\tau = \frac{1}{2}
  \begin{bmatrix}
    I & -q_{j} J \\
    q_{j} J & I
  \end{bmatrix},
\end{equation*}
and so the $j$-th block of $\mathbf{G}$ is given by
\begin{align*}
  \mathbf{G}_{j}(\overline{\bw})
  = M_{j}\, \mathbf{f}_{j}(\overline{\bw}) 
  = \frac{1}{2}
  \begin{bmatrix}
    I & -q_{j} J \\
    q_{j} J & I
  \end{bmatrix}
   \begin{bmatrix}
    0 \\
    -\nabla_{j}E(\br)
   \end{bmatrix}
  = \frac{1}{2}
    \begin{bmatrix}
     q_{j} J\, \nabla_{j}E(\br) \\
    -\nabla_{j}E(\br)
  \end{bmatrix},
\end{align*}
where we write
\begin{equation*}
  \overline{\bw}_{j} =
  \begin{bmatrix}
    \br_{j} \smallskip\\
    q_{j} J\,\br_{j}
  \end{bmatrix}
  \quad\text{for}\quad
  1 \le j \le N.
\end{equation*}
Notice that then $\overline{\bw} \in \mathcal{K}$.
Then we have the following system evolving in $\mathcal{K}$:
\begin{equation}
  \label{eq:massless-fast_time}
  \od{}{T}\overline{\bw} = \eps\,\mathbf{G}(\overline{\bw}),
\end{equation}
which is equivalent to the massless equations~\eqref{eq:Kirchhoff}:
For every $j \in \{1, \dots, N\}$,
\begin{align}
  \label{eq:massless-equivalence}
  \od{}{T}\overline{\bw}_{j} = \eps\,\mathbf{G}_{j}(\overline{\bw})
  \iff
  \od{}{t}\overline{\bw}_{j} = \mathbf{G}_{j}(\overline{\bw})
  \iff
  \dot{\br}_{j} = \frac{q_{j}}{2} J\, \nabla_{j}E(\br).
\end{align}

\subsection{Averaged Massive System and Massless System}
\label{ssec:averaging}
Despite the (slight) difference between the periodic average $\overline{\mathbf{F}}$ of vector field $\mathbf{F}$ and the vector field $\mathbf{G}$ for the massless equations, it turns out that $\mathbf{G}$ is also effectively the periodic average of $\mathbf{F}$, assuming the following:
\begin{assumptions*}
  \noindent
  \begin{enumerate}[(i)]
  \item Solutions $T \mapsto \bw(T)$ to~\eqref{eq:massive-fast_time} are contained in some subset of $T^{*}\R^{2N}$ (more precisely $T^{*} ( \mathbb{B}_{1}(0) )^{N}$) on which $\bw \mapsto \mathbf{F}(\bw,T)$ and $\bw \mapsto \overline{\mathbf{F}}(\bw)$ are Lipschitz.
    \label{assumption-1}
    \smallskip
  \item Massless solution $t \mapsto \br(t)$ to~\eqref{eq:massless-equivalence} is contained in a closed set $\mathcal{D} \subset ( \mathbb{B}_{1}(0) )^{N}$ excluding the inter-vortex collision points $\br_{j} = \br_{k}$ for $j \neq k$ as well as vortex-wall collision points $r_{j} = 1$ for every $j$ (see the definition~\eqref{eq:E} of the potential energy $E$); particularly, this implies that $\overline{\bw} \mapsto \mathbf{G}(\overline{\bw})$ is Lipschitz.
    \label{assumption-2}
  \end{enumerate}
\end{assumptions*}

In what follows, we shall also use the Euclidean norm (2-norm) $\norm{\,\cdot\,}$ and its associated operator norm.

\begin{proposition}[Averaged massive system is approximately massless]
  \label[proposition]{prop:averaging}
  Let $T \mapsto \bw(T)$ and $T \mapsto \overline{\bw}(T)$ be the solutions of~\eqref{eq:massive-fast_time} and~\eqref{eq:massless-fast_time}, respectively, with initial conditions
  \begin{equation*}
    \bw(0) = \overline{\bw}(0) \in \mathcal{K},
  \end{equation*}
  and suppose that both solutions satisfy the above assumptions for every $T \in [0,t_{1}/\eps]$ with some constant $t_{1} > 0$.
  Then, we have
  \begin{equation*}
    \norm{ \bw(T) - \overline{\bw}(T) } \le \eps\, c_{1}\, e^{\lambda_{1} t_{1}}
    \quad
    \forall T \in [0,t_{1}/\eps]
  \end{equation*}
  for some constants $c_{1}, \lambda_{1} > 0$.
\end{proposition}

\begin{proof}
  From~\eqref{eq:massive-fast_time} and~\eqref{eq:massless-fast_time} and the initial conditions, we have
  \begin{align*}
    \bw(T) - \overline{\bw}(T)
    = \eps \int_{0}^{T} ( \mathbf{F}(\bw(\tau), \tau) - \mathbf{G}(\overline{\bw}(\tau)) )\, d\tau.
  \end{align*}
  However, notice that, since $\overline{\bw}(\tau) \in \mathcal{K}$, we have
  \begin{equation*}
    A_{j} \overline{\bw}_{j}(\tau) =
    \begin{bmatrix}
      -q_{j} J & I \\
      -I & -q_{j} J
    \end{bmatrix}
    \begin{bmatrix}
      \br_{j}(\tau) \smallskip\\
      q_{j} J\,\br_{j}(\tau)
    \end{bmatrix}
    = 0
    \implies
    e^{\tau A_{j}} \overline{\bw}_{j}(\tau) = \overline{\bw}_{j}(\tau),
  \end{equation*}
  and thus $e^{\tau A}\,\overline{\bw}(\tau) = \overline{\bw}(\tau)$; hence we have
  \begin{align*}
    \mathbf{G}(\overline{\bw}(\tau))
    &= \frac{1}{2\pi} \int_{0}^{2\pi} e^{-\tau A}\, \mathbf{f}\parentheses{ \overline{\bw}(\tau) } d\tau \\
    &= \frac{1}{2\pi} \int_{0}^{2\pi} e^{-\tau A}\, \mathbf{f}\parentheses{ e^{\tau A}\,\overline{\bw}(\tau) } d\tau \\
    &= \overline{\mathbf{F}}(\overline{\bw}(\tau)).
  \end{align*}
  Therefore, we have
  \begin{align*}
    \mathbf{F}(\bw(\tau), \tau) - \mathbf{G}(\overline{\bw}(\tau))
    &= \mathbf{F}(\bw(\tau), \tau) - \overline{\mathbf{F}}(\overline{\bw}(\tau)) \\
    &= \mathbf{F}(\bw(\tau), \tau) - \mathbf{F}(\overline{\bw}(\tau), \tau) 
      + \mathbf{F}(\overline{\bw}(\tau), \tau) - \overline{\mathbf{F}}(\bw(\tau)),
  \end{align*}
  and so
  \begin{align*}
    \norm{ \bw(T) - \overline{\bw}(T) }
    &\le \eps \int_{0}^{T} \norm{ \mathbf{F}(\bw(\tau), \tau) - \mathbf{F}(\overline{\bw}(\tau),\tau) } d\tau 
      + \eps \norm{
      \int_{0}^{T} \parentheses{ \mathbf{F}(\overline{\bw}(\tau),\tau) - \overline{\mathbf{F}}(\overline{\bw}(\tau)) } d\tau
      }. 
  \end{align*}

  For the first term on the right-hand side, the assumptions imply that $\bw \mapsto \mathbf{F}(\bw,\tau)$ is Lipschitz, i.e., there exists $\lambda_{1} > 0$ such that
  \begin{equation*}
    \norm{ \mathbf{F}(\bw(\tau), \tau) - \mathbf{F}(\overline{\bw}(\tau),\tau) } \le \lambda_{1} \norm{ \bw(\tau) - \overline{\bw}(\tau) }
    \quad
    \forall\tau \in [0, T].
  \end{equation*}
  For the second term, one can show (see the Appendix) that there exists $c_{1} > 0$ such that
  \begin{equation}
    \label{eq:integral_estimate}
    \norm{
      \int_{0}^{T} \parentheses{ \mathbf{F}(\overline{\bw}(\tau),\tau) - \overline{\mathbf{F}}(\overline{\bw}(\tau)) } d\tau
    }
    \le c_{1}
    \quad
    \forall T \in [0, t_{1}/\eps].
  \end{equation}
  Therefore, we have
  \begin{equation*}
    \norm{ \bw(T) - \overline{\bw}(T) }
    \le \eps\, \lambda_{1} \int_{0}^{T} \norm{ \bw(\tau) - \overline{\bw}(\tau) } d\tau 
    + \eps\,c_{1},
  \end{equation*}
  and so Gronwall's inequality~\cite{Gr1919} yields
  \begin{equation*}
    \norm{ \bw(T) - \overline{\bw}(T) }
    \le \eps\, c_{1}\, e^{\eps\,\lambda_{1} T}
    = \eps\, c_{1}\, e^{\lambda_{1} t_{1}}. \qedhere
  \end{equation*}
\end{proof}

One can slightly generalize the above as follows:
\begin{corollary}
  \label[corollary]{cor:averaging}
  Under the same assumptions from~\Cref{prop:averaging} except that the initial conditions satisfy
  \begin{equation*}
    \norm{ \bw(0) - \overline{\bw}(0) } = O(\eps)
    \quad\text{with}\quad
    \overline{\bw}(0) \in \mathcal{K},
  \end{equation*}
  we have
  \begin{equation*}
    \norm{ \bw(T) - \overline{\bw}(T) } \le \eps\,c_{2}\, e^{\lambda_{1} t_{1}}
    \quad
    \forall T \in [0,t_{1}/\eps].
  \end{equation*}
  for some constants $c_{2}, \lambda_{1} > 0$.
\end{corollary}
\begin{proof}
  We now have
  \begin{align*}
    \bw(T) - \overline{\bw}(T)
    = \bw(0) - \overline{\bw}(0) + \eps \int_{0}^{T} ( \mathbf{F}(\bw(\tau), \tau) - \mathbf{G}(\overline{\bw}(\tau)) )\, d\tau.
  \end{align*}
  However, by assumption, there exists $\tilde{c}_{2} > 0$ such that $\norm{ \bw(0) - \overline{\bw}(0) } \le \tilde{c}_{2}\,\eps$, whereas the same argument as above applies to the integral term.
  Therefore,
  \begin{equation*}
    \norm{ \bw(T) - \overline{\bw}(T) }
    \le \eps\, \lambda_{1} \int_{0}^{T} \norm{ \bw(\tau) - \overline{\bw}(\tau) } d\tau 
    + \eps\,(c_{1} + \tilde{c}_{2}),
  \end{equation*}
  Setting $c_{2} \defeq c_{1} + \tilde{c}_{2}$, Gronwall's inequality~\cite{Gr1919} again yields the desired result.
\end{proof}

\subsection{Approximation of Massive System by Massless System}
Using the above proposition, one can approximate solutions of the massive system~\eqref{eq:massive_vec} with $\bz(0)$ being $O(\eps)$-close to $\mathcal{K}$ by corresponding massless solutions:
\begin{theorem}[Approximation of massive system by massless system]
  \label[theorem]{thm:massive-massless}
  Consider the solution $t \mapsto \bz(t)$ of the massive system~\eqref{eq:massive_vec} (or~\eqref{eq:Hamilton}) for $0 \le t \le t_{1}$ with some $t_{1} > 0$ with initial condition satisfying
  \begin{equation*}
    \norm{ \bz(0) - \overline{\bw}(0) } = O(\eps)
    \quad\text{for some}\quad
    \overline{\bw}(0) \in \mathcal{K},
  \end{equation*}
  where
  \begin{equation*}
    \overline{\bw}_{j}(0) =
    \begin{bmatrix}
      \br_{j}(0) \smallskip\\
      q_{j} J \br_{j}(0)
    \end{bmatrix}
    \quad\text{for}\quad
    1 \le j \le N.
  \end{equation*}
  Define $T \mapsto \bw(T)$ by
  \begin{equation}
    \label{eq:w-z}
    \bw(T) \defeq e^{-T A}\,\bz(\eps T),
  \end{equation}
  as well as $T \mapsto \overline{\bw}(T)$ by
  \begin{equation*}
    \overline{\bw}_{j}(T) \defeq
    \begin{bmatrix}
      \br_{j}(\eps T) \smallskip\\
      q_{j} J \br_{j}(\eps T)
    \end{bmatrix}
    \quad\text{for}\quad
    1 \le j \le N,
  \end{equation*}
  where $t \mapsto \br(t)$ is the solution to the massless system~\eqref{eq:Kirchhoff} with the initial point $\br(0)$.

  Then, assuming that $\bw$ and $\overline{\bw}$ satisfy the same assumptions from~\Cref{cor:averaging}, the massless solution
  \begin{equation*}
    t \mapsto \overline{\bw}(t/\eps)
    \quad\text{with}\quad
     \overline{\bw}_{j}(t/\eps) =
    \begin{bmatrix}
      \br_{j}(t) \smallskip\\
      q_{j} J \br_{j}(t)
    \end{bmatrix}
  \end{equation*}
  approximates the massive solution $t \mapsto \bz(t)$ of~\eqref{eq:massive_vec} as follows:
  \begin{equation*}
    \norm{ \bz(t) - \overline{\bw}(t) } \le \eps\,c_{2}\, e^{\lambda_{1} t_{1}}
    \quad
    \forall t \in [0,t_{1}],
  \end{equation*}
  where the constants $c_{2}, \lambda_{1} > 0$ are the same ones from~\Cref{cor:averaging}.
\end{theorem}
\begin{proof}
  Rewriting~\eqref{eq:w-z} in the original time variable $t = \eps\,T$, we have
  \begin{equation*}
    \bz(t) = e^{t A/\eps}\, \bw(t/\eps),
  \end{equation*}
  and so, for every $t \in [0,t_{1}]$,
  \begin{align*}
    \norm{ \bz(t) - \overline{\bw}(t) }
    &= \norm{ e^{t A/\eps}\, \bw(t/\eps) - \overline{\bw}(t/\eps) } \\
    &= \norm{ e^{t A/\eps} \parentheses{ \bw(t/\eps) - \overline{\bw}(t/\eps) } } \\
    &\le \norm{ \bw(t/\eps) - \overline{\bw}(t/\eps) } \\
    &\le \eps\,c_{2}\, e^{\lambda_{1} t_{1}},
  \end{align*}
  where the second line follows because $e^{t A/\eps}\, \overline{\bw}(t/\eps) = \overline{\bw}(t/\eps)$ due to $A \overline{\bw}(t/\eps) = 0$; the third line follows from $\norm{ e^{t A/\eps} } = 1$ due to $A$ being skew-symmetric; the fourth line follows from~\Cref{cor:averaging} noting that $\bw(0) = \bz(0)$ and hence $\norm{ \bw(0) - \overline{\bw}(0) } = O(\eps)$.
\end{proof}

\section{Slow Manifolds and Normal Forms}
\label{sec:slow_manifold}
In the kinematic approximation, $\eps=0$, and the solution remains on the kinematic subspace for all time. When $\eps \ll 1$ is nonzero, the kinematic subspace is not invariant under the flow, but our numerical experiments and averaging result indicate that the solutions starting close to the subspace $\cK$ remain close to it for a significant time and that deviations from $\cK$ take the form of fast oscillations.

A natural question to ask is whether there exists an invariant manifold near $\cK$ on which the fast oscillations are entirely suppressed; this is called a slow manifold $\cS$. As a motivating system, consider an ODE system of the form
\begin{equation}\label{eq:fast_slow}
\begin{split}
\dot{z} &= f(z,w);\\
\eps \dot{w} &= g(z,w)
\end{split}
\end{equation}
for $z\in \mathbb{R}^n$, $w\in \mathbb{R}^m$, and sufficiently smooth functions $f$ and $g$. 
Formally, the $z$ coordinates evolve on a slow $O(1)$ timescale, while the $w$ coordinates evolve on a time scale of $O(\epsilon^{-1})$.
Consider first the kinematic limit $\eps=0$. In a neighborhood of a point $(z_0, w_0)$ at which the matrix of partial derivatives $\pd{g}{w}$ is nonsingular, the implicit function theorem guarantees that we may solve the second equation in the form 
\begin{equation} \label{eq:slow_parameterization}
w=\mathscr{W}(z).
\end{equation}
Solutions confined to this neighborhood obey the kinematic dynamics
$$
\dot{z} = f(z,\mathscr{W}(z)),
$$
which evolves on the $O(1)$ timescale. 

The behavior of the system for small nonzero $\eps$ is of longstanding and continued interest. The question is whether there exists a manifold on which the dynamics of $w$ are ``slaved'' to those of $z$ and so evolve only on the slow time scale. 
When $\eps>0$, we may try to solve for $w$ as a function of $z$ as a formal power series in $\eps$, letting $\cS_j$ define the sequence of $O(\eps^j)$ truncations of this series. Then $\cS$ is the formal limit of this sequence. In similar problems, the answer is often that, while we may find an asymptotic expansion of $\cS$ to arbitrary order in perturbation theory, the constructed manifold is not invariant: transverse oscillations will inevitably arise. We briefly discuss the history of this problem and the mathematical reasons for this non-invariance in the concluding~\Cref{sec:conclusions}. The slow manifold, while not invariant for all times, is approximately invariant, and can be considered as an $\eps$-dependent correction to $\cK$. Despite its non-invariance, it provides an improved approximation to the dynamics over the kinematic approximation.

We derive a \emph{normal form} expansion for the dynamics in a neighborhood of $\cS$ using the Lie transform perturbation method, which can be thought of as a form of averaging. This procedure also generates an approximate expression for $\cS$ itself. Averaging theorems are typically proved by constructing an appropriate near-identity change of variables so that the equations in the transformed coordinates yield simplified equations with rapidly oscillating terms removed. The Lie transform method constructs such a change of variables that is formally canonical, thereby preserving the Hamiltonian form of the equations.

\subsection{Further changes of variables}
To prepare the ground for the normal-form calculation, we introduce new coordinates which separate fast motions along the manifold $\cK$ from slower motions transverse to it. The kinematic subspace is defined as the solution to 
$$
\bp_j= q_j J \br_j.
$$
Guided by this, we define new coordinates
\begin{equation} \label{eq:rpToRP}
\begin{split}
    \bR_j &= \binom{X_j}{Y_j} = \frac{1}{2} \left( \br_j - q_j J \bp_j \right); \\
    \bP_j & = \binom{U_j}{V_j} = \frac{1}{2} \left( \bp_j - q_j J \br_j \right).
\end{split}
\end{equation}
The $\bP_j$ coordinates vanish in the kinematic subspace $\cK$, and thus describe motions transverse to it, while the $\bR_j$ coordinates describe motions within $\cK$. In these coordinates, the symplectic form~\eqref{eq:Omega} is written as
\begin{equation} \label{eq:OmegaXY}
    \Omega = 2q_j
    \parentheses{\d X_j \wedge\d Y_j -\d U_j \wedge\d V_j }
\end{equation}
and the Hamiltonian~\eqref{eq:H} takes the form
\begin{equation} \label{eq:HRP}
    H(\bR,\bP) = H_0(\bP,\eps) + E\left(\bR + \cJ \bP\right),
\end{equation}
where
\begin{equation} \label{eq:calJ}
\cJ= 
  \begin{bmatrix}
     q_1 J & & & & \\ 
     & q_2 J & & & \\ 
     & & \ddots & & \\ 
     & & & & q_N J
  \end{bmatrix}
\end{equation}
is a $(2N)\times(2N)$ matrix and hence $\br = \bR + \cJ \bP$, and
\begin{equation}\label{eq:H0P}
H_0(\bP,\eps) \defeq \frac{2}{\eps} \sum_{j=1}^{N} \parentheses{U_j^2 + V_j^2}.
\end{equation}

Since we assume $\norm{\bP}\ll1$, we expand the vortex-interaction energy as
\begin{equation} \label{eq:E_expansion}
    E(\bR+\cJ\bP) =  \sum_{k=0}^{\infty}E_k(\bR,\bP;\bq), 
\end{equation}
where $E_k(\bR,\bP)$ is a homogeneous polynomial of degree $k$ in the components of $\bP$ with coefficients that depend nonlinearly on $\bR$. The leading term is just
\begin{equation}\label{eq:E0}
    E_0(\bR,\bP;\bq) = E(\bR).
\end{equation}

Considering only the two leading terms of the Hamiltonian~\eqref{eq:HRP}, i.e., $H_0(\bP,\eps)+ E(\bR)$, the dynamics is decoupled into a fast motion in $\bP$ and a slow motion of $\bR$ in the kinematic subspace, which is identical to the evolution of massless vortices. The higher-order terms, of course, couple these motions. 

A final transformation to complex coordinates greatly simplifies the remaining computation. Define
\begin{equation} \label{eq:defZW}
    Z_j = X_j + \rmi Y_j \text{ and } W_j = U_j + \rmi V_j.
\end{equation}
In these variables, the symplectic form is written 
\begin{equation} \label{eq:OmegaZW}
    \Omega = \rmi q_j \left( \d Z_j \wedge\d \bar{Z}_j - \d W_j \wedge\d \bar{W}_j \right),
\end{equation}
where overbars denote complex conjugates,
and the ``fast'' term in the Hamiltonian is
\begin{equation}\label{eq:H0ZW}
    H_0(\bW,\eps) = \frac{2}{\eps}\sum_{j=1}^{N} |W_j|^2.
\end{equation}
The further terms are defined as
\begin{equation}\label{eq:HnZW}
H_k(\bZ,\bW;\bq) \defeq E_{k-1}(\bZ,\bW;\bq),
\end{equation}
where the terms are the right are defined in Eq.~\eqref{eq:E_expansion}, and we note that $H_1=H_1(\bZ;q)$ is independent of $\bW$.

\subsection{Specialization to two-vortex systems}
The normal form computation depends on the null space of a certain linear oscillator that appears in the calculation. The null space, in turn, depends on the dimension of the system and the topological charges of the vortices. From here on, we restrict our attention to two systems of $N=2$ vortices, namely equal co-rotating vortices with $q_1=q_2=1$ and counter-rotating vortices with $-q_1=q_2=1$.

In complex coordinates, the next three terms in series~\eqref{eq:HnZW} are
\begin{subequations} \label{eq:H123ZW}
\begin{equation} \label{eq:H1ZW}
H_1 = E(\bZ) =
\log \left(1-\abs{Z_1}^2\right)+ \log \left(1-\abs{Z_2}^2\right)
+q_1q_2 \log{\frac{\abs{1-Z_1\bar{Z}_2}^2}{\abs{Z_1-Z_2}^2}},
\end{equation}
\begin{equation} \label{eq:H2sub}
H_2 = \left( \frac{q_1 \bar{Z}_1}{1-\abs{Z_1}^2}+\frac{q_2\left(1-\abs{Z_1}^2\right)}{(Z_1-Z_2)(1-Z_1 \bar{Z}_2)}\right)\rmi W_1 
     +\left( \frac{q_2 \bar{Z}_2}{1-\abs{Z_2}^2}-\frac{q_1\left(1-\abs{Z_2}^2\right)}{(Z_1-Z_2)(1-\bar{Z}_1 Z_2)}\right)\rmi W_2 +\cc,
\end{equation}
and
\begin{equation} \label{eq:H3sub}
\begin{split}
H_3 = &-\frac{\abs{W_1}^2}{\left(1-\abs{Z_1}^2\right)^2}
-\frac{\abs{W_2}^2}{\left(1-\abs{Z_2}^2\right)^2} 
+\left(\frac{W_1 W_2}{\left(Z_1-Z_2\right)^2}
-\frac{\bar{W}_1 W_2}{\left(1- \bar{Z}_1 Z_2\right)^2} + \cc\right)\\
&+\frac{W_1^2}{2}  \left(
\frac{\bar{Z}_1^2}{\left(1-\abs{Z_1}^2\right)^2}
-\frac{q_1 q_2 {\left(1-\abs{Z_2}^2\right)^2}
\left(1+\abs{Z_2}^2 -2 Z_1 \bar{Z}_2\right)}{\left(Z_1-Z_2\right)^2 \left(1-Z_1 \bar{Z}_2\right)^2}
\right)+\cc \\
&+\frac{W_2^2}{2}  \left(
\frac{\bar{Z}_2^2}{\left(1-\abs{Z_2}^2\right)^2}
-\frac{q_1 q_2 {\left(1-\abs{Z_1}^2\right)^2}\left(1+\abs{Z_1}^2 -2 \bar{Z}_1 Z_2\right)}{\left(Z_1-Z_2\right)^2 \left(1-\bar{Z}
_1 Z_2\right)^2}
\right)+\cc.
\end{split}
\end{equation}
\end{subequations}
The abbreviation $\cc$ stands for the complex conjugate in the expressions for $H_2$ and $H_3$, the latter of which contains ten monomials when the conjugate terms are included. The series~\eqref{eq:HnZW} is \emph{almost} a power series in the $\bW$ coordinates: $H_0(W,\epsilon)$ is quadratic in $\bW$ but the factor of $\eps^{-1}$ makes it (formally) large; beginning with $H_1$, the term $H_k$ is a homogeneous polynomial of order $(k-1)$ in the $\bW$ components with $\bZ$-dependent coefficients.

\subsection{Hamiltonian normal forms}
We assume that there exists an invariant manifold near the kinematic subspace $\cK$ on which $W_j$ remains small, and its fast dynamics are suppressed. Our goal now is to construct a symplectic transformation to a new coordinate system in which the asymptotic expansion of the slow manifold is parameterized and the evolution equations are simplified.

Given a Hamiltonian $H(x,\eps)$ and a canonical transformation $x=\mathcal{X}(y,\eps)$, we let $G(y,\epsilon)=H(\mathcal{X}(y,\eps),\epsilon)$ represent the Hamilton in the new coordinates. The change of variables $\mathcal{X}(y,\eps)$ that is defined below in Eq.~\eqref{eq:S_evolution} depends on a function $S(x,\epsilon)$ called the Lie transform generating function. To emphasize that the change of variables is constructed using $S$, $G$ is called the \emph{Lie transform} of $H$ with respect to $S$ and is written $G=\mathcal{L}_S H$; see Eq.~\eqref{eq:LieTrans}.

The purpose of the Lie transform perturbation algorithm is to formally construct a canonical change of variables $x = \mathcal{X}(y,\eps)$ that converts a Hamiltonian of the form
\begin{equation}\label{eq:Hsub}
H(x)=H_*(x) = \sum_{j=0}^{\infty} H_j(x)
\end{equation}
to the \emph{normal-form} Hamiltonian 
\begin{equation} \label{eq:Hsuper}
G(y)=H^*(y) = \sum_{j=0}^{\infty} H^j(y).
\end{equation}
In what follows, $x=(\bZ,\bW)$ as defined in Eq.~\eqref{eq:defZW}, and the ``new'' variables are denoted by $y=(\cZ,\cW)$.

The Hamiltonian is in normal form if it is maximally simplified with respect to a criterion that we will define below. In practice, one must truncate the expansions of the two Hamiltonians to some finite order. In that case, the transformation is canonical up to that order but may introduce a small ``non-canonicality'' of higher order into the transformation.  Because $\mathcal{X}(y,\eps)$ is canonical, the symplectic form in the new coordinates takes the same form as in the original coordinates.

\subsection{A preview of the results}
\label{sec:preview}

We first present the results of the calculation, to allow the reader to see the benefit of the calculation before (or instead of) immersing themselves in its details. We first show the leading terms in the normal form Hamiltonian and then formulas for the leading nonzero terms in the slow manifold.

\subsubsection{Preview of the normal-form Hamiltonian}

Given the Hamiltonian $H_*(\bZ,\bW)$ whose leading terms are given by Equations~\eqref{eq:H0ZW} and~\eqref{eq:H123ZW}, the two leading-order terms in the transformed Hamiltonian~\eqref{eq:Hsuper} are unchanged except for changing the variable names:
$$
H^0(\cZ,\cW)=H_0(\cW,\eps) \text{ and } H^1(\cZ,\cW) = H_1(\cZ) = E(\cZ).
$$
The following two terms simplify considerably. Regardless of the choice of topological charges $q_1$ and $q_2$, 
\begin{equation} \label{eq:H2_is_0}
    H^2(\cZ,\cW) \equiv 0.
\end{equation}
Similarly, all terms $H^{2j}$ in the series for the normal form Hamiltonian vanish, i.e., those containing monomials $\cW^{\balpha} \bar{\cW}^{\bbeta}$ where $\abs{\balpha}+\abs{\bbeta}$ is odd, using the multi-index notation defined in Eq.~\eqref{eq:multi-index}.
Thus, the normal form Hamiltonian depends only on even powers of $\abs{\cW}$. 

The form of further terms in the expansion depends on the values chosen for the topological charges $q_j$. In the case of co-rotating vortices $q_1 = q_2 = 1$,  
\begin{equation} \label{eq:H3equal}
H^3_{\rm equal}=
-\frac{\abs{\cW_1}^2}{\left(1-\abs{\cZ_1}^2\right)^2}
-\frac{\abs{\cW_2}^2}{\left(1-\abs{\cZ_2}^2\right)^2}
-\frac{\cW_1 \bar{\cW}_2}{\left(1-\cZ_1 \bar{\cZ}_2\right)^2}
-\frac{\bar{\cW}_1 \cW_2}{\left(1-\bar{\cZ}_1 \cZ_2\right)^2},
\end{equation}
while in the oppositely-rotating case $-q_1 = q_2 = 1$,  
\begin{equation} \label{eq:H3opposite}
H^3_{\rm opposite}=
\frac{\abs{\cW_1}^2}{\left(1-\abs{\cZ_1}^2\right)^2}
-\frac{\abs{\cW_2}^2}{\left(1-\abs{\cZ_2}^2\right)^2}
+\frac{\cW_1 \cW_2}{\left(\cZ_1-\cZ_2\right)^2}
+\frac{\bar{\cW}_1 \bar{\cW}_2}{\left(\bar{\cZ}_1-\bar{\cZ}_2\right)^2}.
\end{equation}

Nonzero higher-order terms in the normal-form series for $H^*$ depend only on products of the four monomials in $\cW$ and $\bar\cW$ that appear in the corresponding expression for $H^3$. While the terms appearing in $H^3_{\rm equal}$ and $H^3_{\rm opposite}$ are subsets of those appearing in $H_3$, this may not be true for higher-order terms. The perturbation procedure, while simplifying the system at a given order, may introduce additional terms at higher order, but these must satisfy the simplicity conditions defining the normal form.

\subsubsection{Preview of the slow manifold expansion}
The particular form of the slow manifold also depends on our choice of the topological charges $q_j$. In the corotating case $q_1=q_2=1$,
\begin{equation} \label{eq:slowWco}
\begin{split}
W_1=& \frac{\rmi \eps}{2} \cdot \frac{-Z_1 Z_2 \bar{Z}_1^2+2 Z_1 Z_2 \bar{Z}_2 \bar{Z}_1-Z_1 \bar{Z}_2-Z_2 \bar{Z}_2+1}{\left(Z_1 \bar{Z}_1-1\right) \left(Z_2 \bar{Z}_1-1\right)  \left(\bar{Z}_1-\bar{Z}_2\right)}  + O(\eps^2);\\
W_2 =& \frac{\rmi \eps}{2} \cdot\frac{Z_1 Z_2 \bar{Z}_2^2-2 Z_1 Z_2 \bar{Z}_1 \bar{Z}_2+Z_1 \bar{Z}_1+Z_2 \bar{Z}_1-1}{\left(\bar{Z}_1-\bar{Z}_2\right) \left(Z_1 \bar{Z}_2-1\right) \left(Z_2 \bar{Z}_2-1\right)}   + O(\eps^2).
\end{split}
\end{equation}
Ignoring the $O(\eps^2)$ truncation error, this defines the approximation $\cS_1$ to the slow manifold; ignoring the terms of $O(\eps)$ gives back $W=0$, i.e., the kinematic subspace $\cK$ corresponds to $\cS_0$, the zeroth-order approximation to the slow manifold.

In the counter-rotating case $-q_1 = q_2 = 1$, the leading-order terms in the slow manifold expansion are
\begin{equation} \label{eq:slowWcounter}
    \begin{split}
        W_1 &= \frac{\rmi \eps}{2} \cdot
        \frac{Z_1 Z_2 \bar{Z}_1^2-2 Z_1 \bar{Z}_1+Z_1 \bar{Z}_2-Z_2 \bar{Z}_2+1}
        {\left(Z_1 \bar{Z}_1-1\right) \left(Z_2 \bar{Z}_1-1\right)
   \left(\bar{Z}_1-\bar{Z}_2\right)}+O(\eps^2);\\
        W_2 & = \frac{\rmi \eps}{2} \cdot
        \frac{Z_1 Z_2 \bar{Z}_2^2-2 Z_2 \bar{Z}_2-Z_1 \bar{Z}_1+Z_2 \bar{Z}_1+1}
        {\left(\bar{Z}_1-\bar{Z}_2\right) \left(Z_1 \bar{Z}_2-1\right) \left(Z_2
   \bar{Z}_2-1\right)}+O(\eps^2).
    \end{split}
\end{equation}
Ignoring the $O(\eps^2)$ truncation error, this defines the approximation $\cS_1$ to the slow manifold.

\subsubsection{Some conclusions}
These formulas tell us a few things right away. Because $\rmi \od{\cW_j}{t}=\pd{H^*}{\bar{\cW}_j}$ and the normal-form expansion of $H^*$ contains no terms linear in $\bar{\cW}_j$, then 
$$
\evalat{\od{\cW_j}{t}}{{\cW=0}}=0,
$$ 
and the subspace $\cW=0$ is invariant under the normal-form dynamics. Since the normal form Hamiltonian on the invariant manifold $\cW=0$ is equivalent to the massless dynamics, the solutions confined to this subspace evolve according to the same equations as massless vortices.
Any change in the slow dynamics due to a nonzero mass, therefore, must arise due to the change of variables from $(\cZ,\cW)$ to $(\bZ,\bW)$ coordinates. However, if $\cW$ is small and nonzero, $H^3$ and subsequent terms couple the dynamics of the $\cZ$ coordinates with those of the $\cW$ coordinates.

Previous work on systems with slow manifolds complicates these conclusions, which we will discuss in the final section and in follow-up work. Our assumption that a slow manifold exists and is approximately invariant can often be made rigorous only for finite times. Terms beyond all algebraic orders in expansion~\eqref{eq:Hsuper} can eventually cause the solution to leave a neighborhood of the slow manifold and may lead to chaotic dynamics~\cite{Camassa.1995, Vanneste.2008}.

\subsection{The Lie Transform Perturbation Algorithm}

We outline the main ideas underlying the Lie transform perturbation algorithm and its implementation. The algorithm was introduced in the 1960s, with significant contributions from Deprit, Henrard, and Hori~\cite{Deprit1969, Henrard1973, Hori1966}. Two readable accounts of the theory underlying this algorithm, as well as its implementation, are the textbook of Meyer et al.\ and the documentation for a software package for computing normal forms by Collins et al.~\cite{Collins2008, Meyer2010}.

The algorithm generates a pair of formal power series in the $\cW$ variables with $\cZ$-dependent coefficients. The first describes the change of variables used to define this \emph{slow manifold}, and the second describes the \emph{normal form} Hamiltonian governing the dynamics in a neighborhood of this manifold. By `formal power series,' we mean that we do not consider the series's convergence. In fact, such series generally diverge, so that any information about the dynamics must come from truncating the series at some finite order. This reflects the fact that while a genuine invariant manifold does not exist, we may construct a manifold to which solutions remain close for a long but finite time. By computing the approximate manifold to higher order, we can further decrease the distance between the manifold and the evolving solutions.
 
In a manner similar to the averaged equations described in~\Cref{ssec:averaging}, the normal form Hamiltonian is useful because it is, ``simpler'' than the original Hamiltonian. Here, ``simple'' has a precise meaning given by Eqs.~\eqref{eq:what_is_simple} and~\eqref{eq:adjoint_resonant} and the paragraph following them. For a fuller technical definition, we refer to~\cite[Theorem 10.1.1]{Meyer2010}.

The algorithm proceeds as follows. Assume that the leading term in the Hamiltonian series is quadratic in $x$, so that the leading-order dynamics satisfies the Hamiltonian linear system
$$
\dot{x} = J\nabla{H_0}(x) = A x.
$$
Then, under the assumption that $A$ is diagonalizable, the transformed Hamiltonian can be constructed such that
\begin{equation} \label{eq:what_is_simple}
    H^j \left(e^{A t} y \right) = H^j(y).
\end{equation}
This is equivalent to the statement that $\PB{H^j}{H_0}=0$. Defining the adjoint operator $\ad_{H_0}f \defeq \PB{H_0}{f}$, we may write this as
\begin{equation}\label{eq:adjoint_resonant}
\ad_{H_0}H^j = 0.
\end{equation}

These two equations provide a mathematical definition of ``simplicity,'' but it may be less clear how this makes the equations simpler to study. What does that mean here? Looking at Eq.~\eqref{eq:H123ZW}, the term $H_j$ is a homogeneous polynomial of degree $j-1$ in the $W$ variables with $Z$-dependent coefficients. There exist $\binom{j+3}{3}$ degree-$j$ monomials in four variables (i.e., $W_1$, $W_2$, $\bar{W}_1$, and $\bar{W}_2$), so $H_j$ is a sum containing as many as $\binom{j+2}{3}$ terms. This is proven using the ``stars and bars'' argument~\cite{Bona.2011, WikipediaStarsAndBars}. You may confirm that the three polynomials shown in Eq.~\eqref{eq:H123ZW} are sums of $\binom{3}{3}=1$, $\binom{4}{3}=4$ and $\binom{5}{3}=10$. By contrast, we will see that the normal form polynomials vanish identically when $j$ is even, and that when $j=2l-1$ is odd, $H_j$ is the sum of at most $l^2$ monomials. Indeed we see that the polynomials~\eqref{eq:H3equal} and~\eqref{eq:H3opposite} each contain four terms. For large values of $j$, the number of monomials in $H_j$ grows cubically, while the number grows  quadratically for $H^j$. Therefore, the normal form is simpler than the original in that $H^j$ contains fewer terms than $H_j$, and in fact vanishes identically when $j$ is even. 

The theorem generalizes when $A$ is not diagonal.

We outline the method here, but omit details for brevity, retaining only enough for the reader to see why certain terms are included in the normal form and others are not; for the full details, consult the references~\cite{Meyer2010,Collins2008}. The method is based on the observation that any symplectic near-identity transformation $x=\mathcal{X}(y,\eps)$ must correspond to the solution of (possibly nonautonomous) Hamiltonian initial value problem
\begin{equation}\label{eq:S_evolution}
\od{x}{\eps} = S(x,\eps); \; x(0) =y.
\end{equation}
The function $S(x,\eps)$ is called the Lie transform generating function or simply the generating function. By implementing the algorithm, we construct the generating function as a formal power series expansion
$$
S(x,\eps) = \sum_{j=0}^\infty \frac{1}{j!} S_{j+1}(x) .
$$
Let $\mathcal{X}_\eps(y)$ be the fundamental solution operator of system~\eqref{eq:S_evolution} with initial condition $y$ at ``time'' $\eps=0$, and $\mathcal{X}(y,\eps) = \mathcal{X}_\eps(y)$ at a fixed value of $\eps$. Then we define the Lie transform of a function $f(x)$ as
\begin{equation}\label{eq:LieTrans}
\mathcal{L}_S(f) \defeq f \circ \mathcal{X}.
\end{equation}

Once we have properly chosen the generating function $S$, the normal-form Hamiltonian is then the Lie transform of the original Hamiltonian:
$$
H^* (y) =\mathcal{L}_S(H_*)(y) = \left(H_* \circ \mathcal{X}\right)(y) = H_*(\mathcal{X}(y,\eps)).
$$
The question remains how to choose the generating function, which will become clear given the algorithm for computing the Lie transform, known as the Deprit triangular array. This provides a structure that enables us to track the numerous terms in the calculation and to construct an algorithm suitable for implementation. We refer the reader to the references above for details omitted below. The array takes the form
$$
\begin{array}{ccccccccc}
\left(H_0 \equiv\right) & H_0^0 & & & & & & &\left(\equiv H^0\right) \\
& \downarrow & \searrow & & & & & \\
\left(H_1 \equiv\right) & H_1^0 & \rightarrow & H_0^1 & & & & &\left(\equiv H^1\right) \\
& \downarrow & \searrow & \downarrow & \searrow & & & \\
\left(H_2 \equiv\right) & H_2^0 & \rightarrow & H_1^1 & \rightarrow & H_0^2 & & &\left(\equiv H^2\right)\\
& \downarrow & \searrow & \downarrow & \searrow & \downarrow & \searrow & \\
\left(H_3 \equiv\right) & H_3^0 & \rightarrow & H_2^1 & \rightarrow & H_1^2 & \rightarrow & H_0^3 &\left(\equiv H^3\right)\\
& \vdots & \ddots & \vdots & \ddots & \vdots & \ddots & \vdots
\end{array}
$$
A relationship defined by an arrow like $a \rightarrow b$ indicates the direction of dependence, i.e., that constructing $b$ depends on knowledge of $a$. Each intermediate term $H_j^i$,  is defined as a sum of Poisson brackets of terms in the column to its left and above it in the array with appropriate terms in the generating function series, 
\begin{equation} \label{eq:Lie_Poisson_S1um}
H_j^{i} \defeq H_{j+1}^{i-1}+\sum_{k=0}^j\binom{j}{k}\left\{H_{j-k}^{i-1}, S_{k+1}\right\}.
\end{equation}

To compute $H^j=H_0^j$ for $j\ge 1$, we begin with $H_j=H_j^0$ and work from left to right. At the beginning of this step, the expansion of $S$ up to $S_{j-1}$ are known, so that computing all the intermediate terms involves explicit computations in terms of previously-computed terms.

Computing the rightmost term in the row $H_0^j$ involves the yet-to-be-determined expression $S_j$. Our goal is then to choose $S_j$ in order to make $H^j$ as simple as possible. After some computations (not shown), it can be combined into an equation relating $H^j$ and $H_j$
\begin{equation}\label{eq:HtoH}
H^j =H_0^j= \PB{H_0}{S_j} + H_j + \mathcal{B}_j = \ad_{H_0}S_j + H_j + \mathcal{B}_j,
\end{equation}
where $\mathcal{B}_j$ terms can be considered ``byproducts'' produced by the algorithm up to this point.

Our goal is to choose $S_j$ in a form that simplifies $H^j$ as much as possible. The simplest possible form would be $H^j=0$, which would require
\begin{equation}\label{eq:homological_complete}
\ad_{H_0} S_j = - H_j - \mathcal{B}_j.
\end{equation}
This would require the right-hand side of this equation to be in the range of the operator $\ad_{H_0}$, which may or may not hold.  At each order of the perturbation expansion, we will identify a finite-dimensional vector space $\mathbb{V}$ such that $\ad{H_0}:\mathbb{V}\to\mathbb{V}$. Once we more precisely define $\mathbb{V}$, then we may appeal to the Fredholm alternative to decompose
\[
\mathbb{V} = \Ran{\ad_{H_0}} \oplus \Null{\ad_{H_0}},
\]
where we have used the fact that $\ad_{H_0}=\ad_{H_0}^T$ when $A = J\nabla H_0$ is diagonalizable. The definition of the transpose operator depends on the inner product defined below in Eq.~\eqref{eq:ipHN}.

We must therefore consider the space on which this operator acts. First, we introduce the space
\begin{equation}\label{eq:Hn}
\mathbb{H}_n(\bW,\bar{\bW}; \bZ,\bar{\bZ})
:= \left\{ 
   \sum_{\abs{\balpha}+\abs{\bbeta}= n}
   a_{\balpha,\bbeta}(\bZ,\bar{\bZ})\,
   \bW^{\balpha} \bar{\bW}^{\bbeta}
   \;{\Bigg|}\;
   a_{\balpha,\bbeta}\in C^\omega(\bZ,\bar{\bZ})
   \right\}.
\end{equation}
of polynomials of homogeneous degree $n$ in the $(\bW,\bar{\bW})$ coordinates with coefficients that are real-analytic functions of the $(\bZ,\bar{\bZ})$ coordinates. 
Here $\balpha$ and $\bbeta$ are multi-indices, i.e., $N$-tuples of non-negative integers. In this section, we use the multi-index definitions
\begin{equation}\label{eq:multi-index}
\begin{gathered}
\balpha=\left(\alpha_1,\ldots,\alpha_N\right) \text{ for }
\alpha_j \in \mathbb{Z}_{\ge 0}; \\
\abs{\balpha}=\sum_{j=1}^N \alpha_j; \quad
\bW^{\balpha}= \prod_{j=1}^N W_j^{\alpha_j}; \quad
\partial_{\bW}^{\balpha}=\partial_{W_1}^{\alpha_1}\cdots\partial_{W_N}^{\alpha_N}.
\end{gathered}
\end{equation}
The natural inner product on this space is 
\begin{equation}\label{eq:ipHN}
    \ip{F}{G}_{\mathbb{H}^N} \defeq \evalat{F(\partial_\bW,\partial_{\bar{\bW}})G(\bW,\bar{\bW})}{\bW=\bar{\bW}=\mathbf{0}}.
\end{equation}
The inner product of a pair of monomials is 
$\ip{\bW^{\balpha}\bar{\bW}^{\bbeta}}{\bW^{\balpha'}\bar{\bW}^{\bbeta'}}_{\mathbb{H}^N}
= \delta_{\balpha \balpha'} \delta_{\bbeta \bbeta'} \balpha! \bbeta!$, where the multi-index factorial is $\balpha!=\prod_{j=1}^N \alpha_j!$.

Then we introduce the space of polynomials of degree at most $n$ in the $(\bW,\bar{\bW})$ coordinates with coefficients that are real-analytic functions of the $(\bZ,\bar{\bZ})$ coordinates and
\begin{align*}
\mathbb{P}_n(\bW,\bar{\bW}; \bZ,\bar{\bZ})
&:= \left\{ 
   \sum_{\abs{\balpha}+\abs{\bbeta}\le n}
   a_{\balpha,\bbeta}(\bZ,\bar{\bZ})\,
   \bW^{\balpha} \bar{\bW}^{\bbeta}
   \;\Bigg|\;
   a_{\balpha,\bbeta}\in C^\omega(\bZ,\bar{\bZ})
   \right\} 
\\ 
&= \bigoplus_{k=0}^n 
\mathbb{H}_k(\bW,\bar{\bW}; \bZ,\bar{\bZ}).
\end{align*}
The inner product definition~\eqref{eq:ipHN} can be extended to $\mathbb{P}_n$ and vanishes identically when applied to a pair of monomials of different orders.

The series for $H_*$ is composed of terms
\[
H_j \in \mathbb{H}_{j-1}(\bW,\bar{\bW}; \bZ,\bar{\bZ}).
\]

Then the Poisson bracket defined using the symplectic form~\eqref{eq:OmegaZW} satisfies
\[
\PB{\cdot}{\cdot}:\mathbb{P}_m(\bW,\bar{\bW}; \bZ,\bar{\bZ})\times \mathbb{P}_n(\bW,\bar{\bW}; \bZ,\bar{\bZ}) \to \mathbb{P}_{m+n-2}(\bW,\bar{\bW}; \bZ,\bar{\bZ}).
\]
More importantly, because $H_0$ is independent of $(\bZ,\bar{\bZ})$ and quadratic in $(\bW,\bar{\bW})$,
\[
\ad_{H_0}: \mathbb{H}_n(\bW,\bar{\bW}; \bZ,\bar{\bZ}) \to \mathbb{H}_n(\bW,\bar{\bW}; \bZ,\bar{\bZ}),
\]
and
\[
\ad_{H_0}: \mathbb{P}_n(\bW,\bar{\bW}; \bZ,\bar{\bZ}) \to \mathbb{P}_n(\bW,\bar{\bW}; \bZ,\bar{\bZ}).
\]
These two statements can be verified via direct calculation. They allow us to define
\[
\ad_{H_0}^{(n)} \defeq 
\evalat{\ad_{H_0}}{\mathbb{P}_n}.
\]

From these statements, we may verify by mathematical induction that the byproduct terms that appear in Eq.~\eqref{eq:HtoH} satisfy $\mathcal{B}_j \in \mathbb{P}_{j-1}(\bW,\bar{\bW}; \bZ,\bar{\bZ})$. Thus, at each stage of the perturbation expansion, we decompose
\[
H_j + \mathcal{B}_j = \Res_j + \Non_j, \text{ where } 
\Res_j \in \Null\parentheses{\ad_{H_0}^{(j-1)}} \text{ and }
\Non_j \in \Ran\parentheses{\ad_{H_0}^{(j-1)}}.
\]
The terms in $\Non_j$ are called \emph{nonresonant}. By an appropriate choice of terms that form $S_j\in \mathbb{P}_{j-1}(\bW,\bar{\bW}; \bZ,\bar{\bZ})$, we may ensure that no terms of this type appear in $H^j$. The terms that form $\Res_j$ are called \emph{resonant}. By contrast, there does not exist a generating function $S_j$ that can remove such terms from $H^j$.
Since $\mathbb{P}_{j-1}$ is finite-dimensional, so are these subspaces, implying that the terms in Eq.~\eqref{eq:homological_complete} can be interpreted as finite-dimensional vectors and matrices.

We now apply the above general discussion and describe how to arrive at the slow manifold normal form presented in \Cref{sec:preview}. We wrote the Hamiltonian in terms of the complex coordinates given in Eq.~\eqref{eq:defZW} because the operator $\ad_{H_0}$ acts diagonally on monomials in $(\bW,\bar{\bW})$. In particular,
$$
\ad_{H_0}^{(j-1)} {\bW}^{\balpha}\bar{\bW}^{\bbeta} = \frac{2\rmi \bq\cdot(\balpha-\bbeta)}{\eps} {\bW}^{\balpha}{\bar{\bW}}^{\bbeta}.
$$
Clearly, then, the condition that a monomial be resonant is 
\begin{equation} \label{eq:resonance_condition}
{\bW}^{\balpha}{\bar{\bW}}^{\bbeta} \in \Null(\ad_{H_0}) \Leftrightarrow \bq\cdot(\balpha-\bbeta)=0.
\end{equation}
Therefore, we can deal with the nonresonant terms as follows. Write
\[
\Non_j = \sum_{(\balpha,\bbeta)\in \cR_{j-1}} c_{\balpha,\bbeta} {\bW}^{\balpha} {\bar\bW}^{\bbeta}
\text{ and }
S_j = \sum_{(\balpha,\bbeta)\in \cR_{j-1}} a_{\balpha,\bbeta} {\bW}^{\balpha} {\bar\bW}^{\bbeta},
\]
where
\[
\cR_{j-1} \defeq \left\{ \left. (\balpha,\bbeta) \right\rvert \bW^{\balpha} {\bar{\bW}}^{\bbeta} \in \Ran\left(\ad_{H_0}^{(j-1)} \right)\right\},
\]
then
\begin{equation} \label{eq:homological_solution}
a_{\balpha,\bbeta} = \frac{\eps c_{\balpha,\bbeta}}{2\rmi \bq\cdot(\balpha-\bbeta)}.
\end{equation}

The terms satisfying the resonance condition~\eqref{eq:resonance_condition} cannot be removed by a change of variables and make up the normal form Hamiltonian. 
Any $\bW$-independent function $F(\bZ)$ is resonant, since it lies in the adjoint null-space of $\ad_{H_0}$.

\subsection{Deriving the changes of variables}

The above procedure generates a formal power series for a transformed Hamiltonian from which all resonant terms have been removed. Knowledge of this transformed Hamiltonian suffices for many applications. Here, however, our goal is to derive a series for a slow manifold on which oscillations in the direction of the fast degree of freedom are suppressed. The Lie transform~\eqref{eq:LieTrans} provides a way to transform any function once the generating function $S$ is known. In particular, if $f(\cZ,\cW)$ is any component of the identity operator, then we may use Eq.~\eqref{eq:Lie_Poisson_S1um} to evaluate its image under the Lie transform, which tells us how the coordinates transform under this mapping. Collins et al.\ discuss this in detail~\cite{Collins2008}.

\subsection{Application to a system of two vortices}

We restrict our calculation to the two cases discussed above: the corotating case of two identical vortices and the counter-rotating case of two vortices of equal strength but opposite sign.

In either case, and in fact for any system of point vortices, the term $H_1 = E_0(\bZ;\bq)$, as defined in Eq.~\eqref{eq:HnZW}, is $\bW$-independent, and thus resonant, so that 
\[
S_1(\bZ,\bW) = 0 \text{ and } H^{1}(\bZ;\bq) = H_1(\bZ;\bq).
\]

Collins gives explicit formulas for the change-of-variables defined by the Lie transform in terms of iterated Poisson brackets. These simplify considerably when $S_1$ vanishes. Let
\[
y = \left(
\cZ_1, \cZ_2, \cW_1, \cW_2, \bar{\cZ}_1, \bar{\cZ}_2, \bar{\cW}_1, \bar{\cW}_2
\right)
\]
be the vector of ``new'' variables
and $\mathcal{X}_k(y)$ be the $k$th component of the map $\mathcal{X}(y,\eps)$ from new variables to the old. Then we may show that
\[
\mathcal{X}_k(y) = y_k + \frac{1}{2}\PB{y_k}{S_2} + \frac{1}{6}\PB{y_k}{S_3} + \ldots
\]
truncated to the same order as the approximations~\eqref{eq:H3equal} and~\eqref{eq:H3opposite}. Below, we present the expansion up to the first correction term.

\subsubsection*{The corotating case}

In the corotating case, $q_1=q_2=1$, the resonance condition~\eqref{eq:resonance_condition} becomes
$$
\alpha_1 + \alpha_2 - \beta_1 - \beta_2 = 0.
$$
If the monomial is in $\mathbb{H}_j$, then
$$
\abs{\balpha}+\abs{\bbeta} = \alpha_1 + \alpha_2 + \beta_1 + \beta_2 = j.
$$
Combining these two conditions implies that
$$
2(\alpha_1 + \alpha_2) = j,
$$
which implies that the space of resonant monomials for any odd degree $j$ is empty. This, in turn, implies that the normal form for the slow manifold contains only monomials of even total degree. For $j=2$, we find that the system has four solutions,
$$
\parentheses {\balpha , \bbeta } \in 
\braces{ 
\parentheses{ \binom{1}{0}, \binom{1}{0}},
\parentheses{ \binom{0}{1}, \binom{0}{1}},
\parentheses{ \binom{1}{0}, \binom{0}{1}},
\parentheses{ \binom{0}{1}, \binom{1}{0}}
 },
$$
so that
$$
\Res_2 = \Span\{W_1 \bar{W}_1, W_2 \bar{W_2}, W_1 \bar{W}_2, \bar{W}_1 W_2 \}
$$
in agreement with the first nonzero term in the normal form Hamiltonian expansion displayed in Eq.~\eqref{eq:H3equal}. All resonant terms of higher degree are formed from products of these four monomials.

The change of variables takes the form
\begin{equation}\label{eq:ZWfromcZcW}
\begin{split}
Z_1 &=\cZ_1 + \frac{\rmi \eps}{2} \left(-\tfrac{\cW_1}{\left(\cZ_1 \bar{\cZ}_1-1\right)^2}+\bar{\cW}_1 \left(-\tfrac{\cZ_1^2}{\left(\cZ_1 \bar{\cZ}_1-1\right)^2}-\tfrac{\cZ_2^2}{\left(\cZ_2
   \bar{\cZ}_1-1\right)^2}+\tfrac{1}{\left(\bar{\cZ}_1-\bar{\cZ}_2\right)^2}\right)-\tfrac{\cW_2}{\left(\cZ_2
   \bar{\cZ}_1-1\right)^2}-\tfrac{\bar{\cW}_2}{\left(\bar{\cZ}_1-\bar{\cZ}_2\right)^2} \right)\\
   &\quad+O(\eps \abs{\cW}^2)+ O(\eps^2 \abs{\cW});\\
Z_2 &= \cZ_2 +\frac{\rmi \eps}{2}\left(-\tfrac{\cW_1}{\left(\cZ_1 \bar{\cZ}_2-1\right)^2}+\bar{\cW}_2 \left(-\tfrac{\cZ_1^2}{\left(\cZ_1
   \bar{\cZ}_2-1\right)^2}+\tfrac{1}{\left(\bar{\cZ}_1-\bar{\cZ}_2\right)^2}-\tfrac{\cZ_2^2}{\left(\cZ_2
   \bar{\cZ}_2-1\right)^2}\right)-\tfrac{\bar{\cW}_1}{\left(\bar{\cZ}_1-\bar{\cZ}_2\right)^2}-\tfrac{\cW_2}{\left(\cZ_2 \bar{\cZ}_2-1\right)^2}\right) \\
   &\quad+O(\eps \abs{\cW}^2) + O(\eps^2 \abs{\cW});\\
W_1&= \cW_1 +\tfrac{\rmi \eps}{2} \tfrac{-\cZ_1 \cZ_2 \bar{\cZ}_1^2+2 \cZ_1 \cZ_2 \bar{\cZ}_2 \bar{\cZ}_1-\cZ_1 \bar{\cZ}_2-\cZ_2 \bar{\cZ}_2+1}{\left(\cZ_1 \bar{\cZ}_1-1\right) \left(\cZ_2 \bar{\cZ}_1-1\right)  \left(\bar{\cZ}_1-\bar{\cZ}_2\right)} +O(\eps \abs{\cW}) + O(\eps^2);\\
W_2 &= \cW_2 +\tfrac{\rmi \eps}{2} \tfrac{\cZ_1 \cZ_2 \bar{\cZ}_2^2-2 \cZ_1 \cZ_2 \bar{\cZ}_1 \bar{\cZ}_2+\cZ_1 \bar{\cZ}_1+\cZ_2 \bar{\cZ}_1-1}{\left(\bar{\cZ}_1-\bar{\cZ}_2\right) \left(\cZ_1 \bar{\cZ}_2-1\right) \left(\cZ_2
   \bar{\cZ}_2-1\right)}  + O(\eps \abs{\cW}) + O(\eps^2);
\end{split}
\end{equation}
The $O(\eps^2)$ terms in the equations for $W_1$ and $W_2$ are $\cW$-independent.

The next term in the normal form Hamiltonian expansion~\eqref{eq:H3equal} shows us that the hyperplane $\cW=0$ is invariant under the normal form dynamics truncated to this order. Since all higher-order terms take the same form, this hyperplane is invariant under the normal form dynamics truncated to any order. Therefore, the slow manifold is given by
\[
(\bZ,\bW) = \mathcal{X}(\cZ,\cW,\eps).
\]
Plugging $\cW=0$ into the first two equations of~\eqref{eq:ZWfromcZcW} gives
\begin{equation} \label{eq:cZtoZ}
Z_1 = \cZ_1; \; Z_2 = \cZ_2,
\end{equation}
so we may write $W_1$ and $W_2$ directly in terms of $Z_1$ and $Z_2$ as shown in~Eq.~\eqref{eq:slowWco}.

\subsubsection*{The counter-rotating case}
If $q_1=-q_2=1$, the resonance condition becomes
$$
\alpha_1 - \alpha_2 - \beta_1 + \beta_2 = 0.
$$
All resonant monomials are again of even degree $j$ and must also satisfy
$$
2(\alpha_1 + \beta_2) = j.
$$
For $j=2$, we find that
$$
\Res_2 = \Span\{W_1 \bar{W}_1, W_2 \bar{W_2}, W_1 W_2, \bar{W}_1 \bar{W}_2 \}
$$
in agreement with the first nonzero term in the normal expansion for $H^*$, Eq.~\eqref{eq:H3opposite}. Again, all resonant terms of higher degree consist of products of these four monomials. The leading order terms in the expansion of the slow manifold are shown in Eq.~\eqref{eq:slowWcounter}.

\subsection{Numerical results on $\cS_1$}
\Cref{fig:t-x01_on_K} depicts dynamics of the $x$ coordinate of the vortex with topological charge $q_1=-1$ in a system of two counter-rotating vortices. It shows two time series that are largely similar but display key differences. The first is the motion of a vortex in a system of two massless vortices, and the second shows the motion of one in the massive system with initial conditions given on the kinematic manifold $\cK$. The two time series gradually go out of phase over the course of the simulation. The first also undergoes an additional oscillation of small amplitude and high frequency; this is not visible on the scales shown, but we plot it in~\Cref{fig:t-x01_on_S1_closeup}, which we discuss below.

We perform an additional numerical simulation of the massive system~\eqref{eq:Hamilton}, this time with initial data constructed to lie on $\cS_1$, the first nonlinear approximation to the slow manifold $\cS$. To find these, we begin with the initial values in Eq.~\eqref{eq:IC_on_K}, and use these to construct the $\bR_j$ and $\bP_j$ coordinates using Eq.~\eqref{eq:rpToRP}, which give $\bP_j\equiv0$. We complexify $\bR_j$ to compute $\bZ$, and finally compute $\bW$ using Eq.~\eqref{eq:slowWcounter}. Inverting this sequence of transformations gives us an initial condition in the original coordinate system:
\begin{equation}
  \label{eq:IC_on_S1}
  \begin{array}{c}
    \br_{1}(0) = (x_{1}(0), y_{1}(0)) = (0.602543, 0.200491),\quad \br_{2}(0) = (x_{2}(0), y_{2}(0)) = (-0.300376, -0.401214), \medskip\\
    \bp_{1}(0) = (-0.199509, 0.597457),\quad \bp_{2}(0) =  (-0.398786, 0.299624).
  \end{array}
\end{equation}
By construction, these initial conditions do not satisfy Eq.~\eqref{eq:kinematic_constraint}, so do not lie on $\cK$.

We plot the results in~\Cref{fig:IC_on_S1}. In the view shown in~\Cref{fig:t-x01_on_S1}, the two simulations of the massive-vortex system are indistinguishable. In~\Cref{fig:t-x01_on_S1_closeup}, which shows the solution over a much shorter time interval, we can see that the solution that starts on $\cK$ oscillates about the solution that starts on $\cS_1$. 
\begin{figure}[htbp]
  \centering
  \begin{subcaptionblock}[b]{0.475\textwidth}
    \centering
    \includegraphics[width=\textwidth]{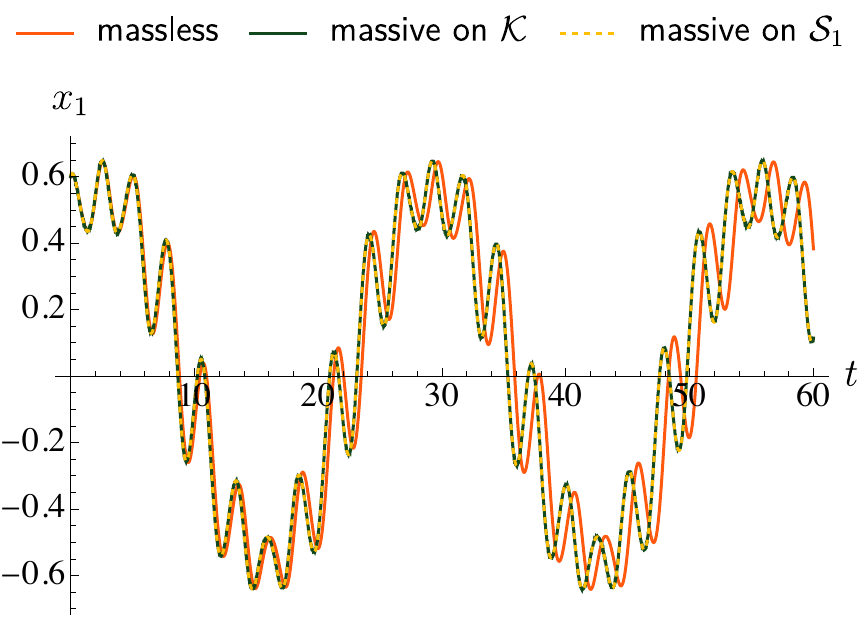}
    \caption{$x$-coordinate of vortex~1}
    \label{fig:t-x01_on_S1}
  \end{subcaptionblock}
  \hfill
  \begin{subcaptionblock}[b]{0.475\textwidth}
    \centering
    \includegraphics[width=\textwidth]{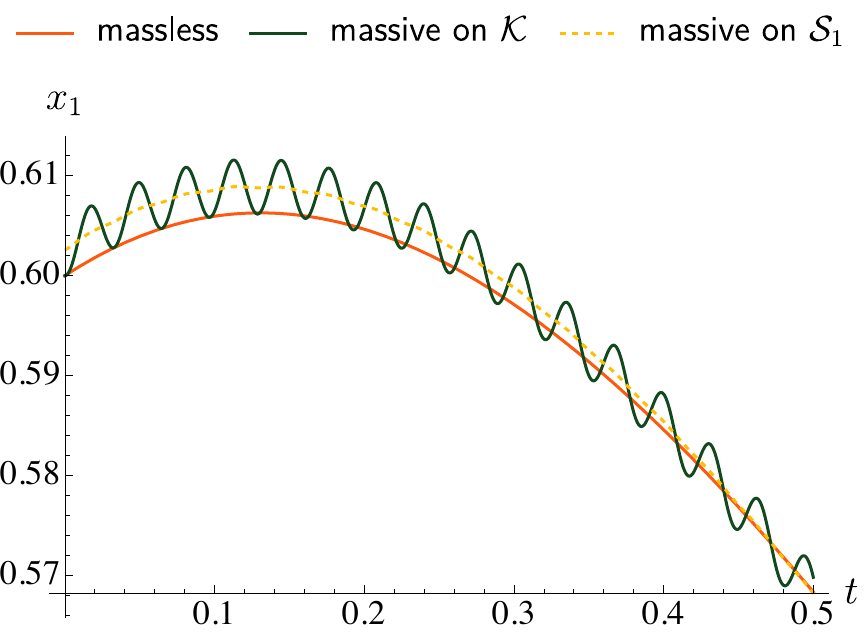}
    \caption{$x$-coordinate of vortex~1---shorter time}
    \label{fig:t-x01_on_S1_closeup}
  \end{subcaptionblock}
  \caption{Massless and massive time series of vortex~1 with parameters and initial conditions from~\eqref{eq:IC_on_K} satisfying $(\br(0),\bp(0)) \in \cK$ and a third with initial condition from~\eqref{eq:IC_on_S1}.}
  \label{fig:IC_on_S1}
\end{figure}

Choosing the initial condition on $\cS_1$ reduces but does not eliminate the fast oscillations. \Cref{fig:t-V1}, shows two plots of $V_1(t) = (\xi_1(t) - x_1(t))/2$ (the second component of $\bP_1 = \frac{1}{2} \left( \bp_1 - q_1 J \br_1 \right)$ where $\bp_1 = (\xi_1,\eta_1)$), the first with initial condition~\eqref{eq:IC_on_K} on $\cK$ and the second with initial condition~\eqref{eq:IC_on_S1} on $\cS_1$, both over a short timescale.
On the kinematic subspace $V_1\equiv0$, so the solution that starts on $\cK$ has $V_1(0)=0$, but the fast oscillation is large relative to the trajectories whose initial condition lies on $\cS_1$. By computing more terms in the slow manifold expansion, we can reduce the amplitude of the fast oscillations further.

\begin{figure}[htbp]
  \centering
  \begin{subcaptionblock}[b]{0.475\textwidth}
    \centering
    \includegraphics[width=\textwidth]{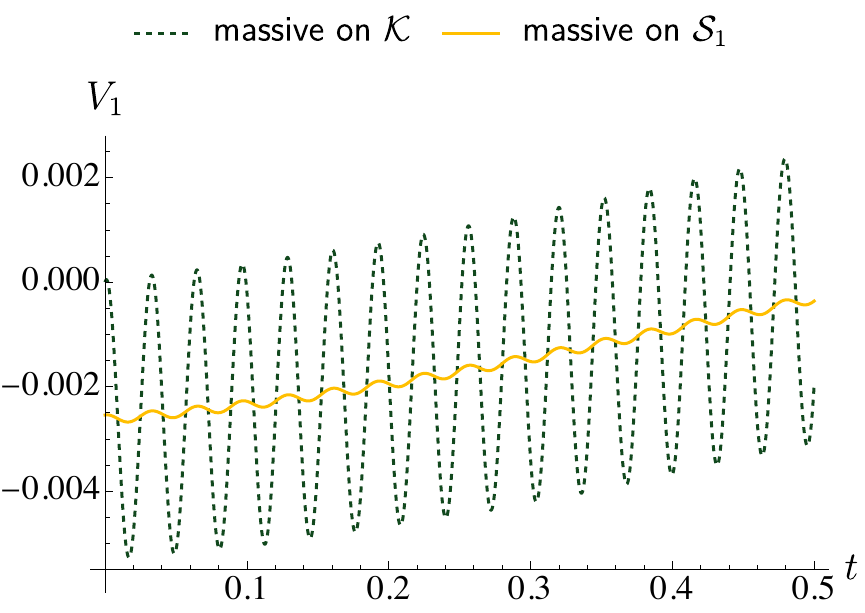}
    \caption{$V$-coordinate of vortex~1}
    \label{fig:t-V1}
  \end{subcaptionblock}
  \hfill
  \begin{subcaptionblock}[b]{0.475\textwidth}
    \centering
    \includegraphics[width=\textwidth]{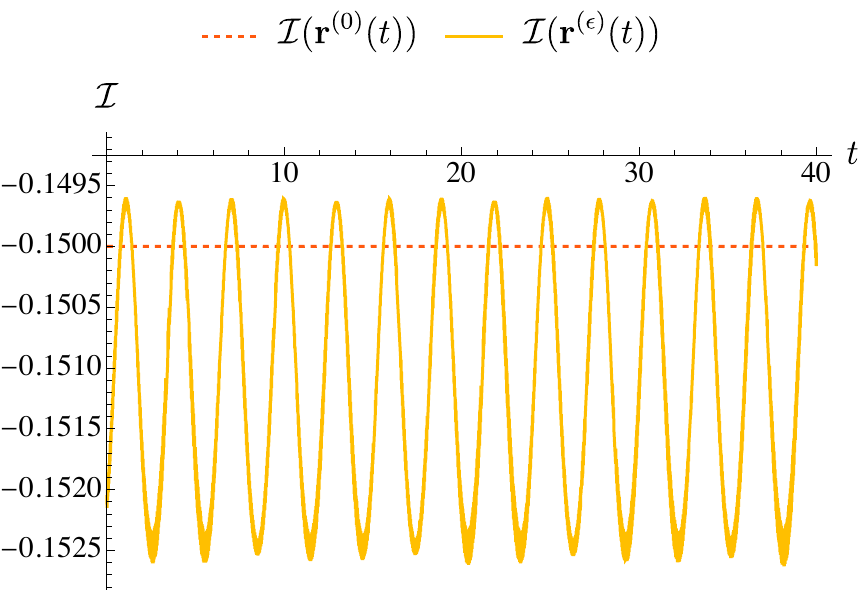}
    \caption{Angular impulse.}
    \label{fig:AngImpulse_on_S1}
  \end{subcaptionblock}
  \caption{(a) The $V_1$ component of the motion transverse to $\cK$ for the initial condition~\eqref{eq:IC_on_K} on $\cK$ and~Eq.~\eqref{eq:IC_on_S1} on $\cS_1$. (b) Time evolution of  angular impulse $\mathcal{I}$ (see~\eqref{eq:AngImpulse}) along  massless solution $\br^{(0)}(t)$ and massive solution $\br^{(\epsilon)}(t)$ with parameters and initial conditions from~\eqref{eq:IC_on_S1} on $\cS_1$. }
  \label{fig:IC_on_S1_b}
\end{figure}

\Cref{fig:AngImpulse_on_S1} shows that the fast oscillations in the value of the massless angular impulse seen in~\Cref{fig:AngImpulse_on_K} are significantly reduced when the initial condition is chosen on $\cS_1$. Slower oscillations remain.

\section{Conclusions and Outlook}
\label{sec:conclusions}
\subsection{Conclusions}
In this work, we analyzed the dynamics of massive point vortices in immiscible Bose--Einstein condensates in the small core mass limit $\eps\to 0$.
We showed that the well-known Kirchhoff equations governing massless vortex dynamics emerge as the leading-order approximation to the massive equations.

Our key findings are:
\begin{enumerate}[(i)]
  \item The identification of a subspace $\cK$ in the phase space of the massive dynamics such that the massive dynamics starting $O(\eps)$-close to $\cK$ stays $O(\eps)$-close to the corresponding massless dynamics for short times;
  \item A rigorous proof of this near-massless behavior of the massive dynamics;
  \item Computation of the leading terms in the approximations of the slow manifold $\cS$ and a near-identity change of variables into a normal form that decouples the fast inertial oscillations from a slow motion equivalent to the massless dynamics.
  \item Numerical computation showing the effectiveness of the normal form dynamics in a neighborhood of $\cS$.
\end{enumerate}

From a physical perspective, our results are particularly relevant to many experimental systems because vortices acquire an effective inertial mass $\eps$ due to their tendency to trap particles (tracers, impurities, quantum and thermal excitations, etc).
As a result, the conventional massless point-vortex model overlooks an essential feature of real systems.
By systematically analyzing the asymptotic behaviors as $\eps \to 0$, our work provides a more realistic and physically meaningful description of vortex motion in quantum fluids. We expect it to give a better understanding of various recent and ongoing experiments~\cite{Kwon2021, Reeves2022, Hernandez2024, Neely2024}. 

\subsection{Outlook}
This study suggests several directions for future research.
Thus far, we have discussed only the similarity between the motions of massive and massless vortices, despite several key differences between them. For example, massive vortices in a hard-walled trap can collide with the boundary in finite time, an effect not possible for massless vortices~\cite{Caldara2023}. More relevant to the current study, the dynamics on the slow manifold $\cS$ are (formally) identical to the massless dynamics, which are completely integrable for a system of $N=2$ vortices. However, the introduction of mass increases the dimension of the phase space without increasing the number of independent conserved quantities, breaking the system's integrability. This should allow for a wider set of solution behaviors, including chaotic dynamics, especially for solutions on longer time intervals. To leading order, these behaviors should be captured by the next-order terms in the Hamiltonian expansions; see Equations~\eqref{eq:H3equal} and~\eqref{eq:H3opposite}, which couple the dynamics of the $\cZ$ coordinates in the slow manifold to the $\cW$ coordinates transverse to it. Understanding this will require a more detailed study of the massless dynamics before the effects of mass are introduced.

\Cref{fig:IC_on_S1} shows that solutions starting on $\cS_1$, the first nontrivial approximation to the slow manifold $\cS$, remain near the slow manifold for the time scale simulated. However, the normal-form calculation of~\Cref{sec:slow_manifold} is insufficient to prove the existence of a slow manifold, since the series produced by the algorithm usually has zero radius of convergence. This is a longstanding concern in dynamical systems, which we briefly summarize. The main ideas were developed in the study of two model problems, the first arising in the dynamics of atmospheric waves, and the second in the motion of electrons in the presence of curved magnetic fields. In both of these systems, as for point vortices with small mass, the motion transverse to the slow manifold is oscillatory. The other main case, when the slow manifold is saddle-like, is better understood; the existence and stability of such slow manifolds goes back to work of Fenichel~\cite{Fenichel:1979}.

Lorenz and Krishnamurthy proposed a two-degree-of-freedom model of the atmosphere with one degree of freedom representing (slow) Rossby waves and a second representing (fast) gravity waves, and showed numerically that fast oscillations invariably arose, despite careful choice of initial conditions~\cite{Lorenz:1987}. Camassa used a Melnikov integral argument to identify the mechanism of energy transfer from slow degrees of freedom to fast~\cite{Camassa.1995}. He noted that his argument suffices to construct a Smale horseshoe, thereby proving the existence of chaotic dynamics. Subsequently, Vanneste derived leading-order asymptotics for terms in the tail of the slow manifold expansion and, using Borel summation, found beyond-all-orders terms in the asymptotics, showing that the slow manifold is only piecewise analytic; when trajectories cross a Stokes surface in the manifold, they excite exponentially small oscillations transverse to $\cS$~\cite{Vanneste.2008}.

Another promising avenue is a more detailed investigation of the so-called \emph{plasma-orbit theory}, which describes vortex trajectories through an analogy with the motion of charged particles in an external electromagnetic field; see \cite{Fischer1999, Caldara2023, Poli2024}. Questions about slow manifolds have arisen frequently in the study of plasma-orbit theory, where it is often referred to as \emph{guiding center motion}~\cite{Cary:2009, Littlejohn:1979, Littlejohn:1981, Littlejohn:1983}. These references primarily concern the derivation of the reduced Hamiltonian equations. Burby and Hirvijoki recently proved the existence and stability of a slow manifold under additional assumptions on the reduced Hamiltonian and with a single fast degree of freedom~\cite{Burby:2021}, as opposed to the present case, where each massive vortex contributes a fast degree of freedom.

In \Cref{sec:slow_manifold}, we find that the higher-order terms in the Hamiltonian and the form of the slow manifold both depend on the choice of signs of the topological charges $q_j$. This is because the form of the resonance depends on the set of exponents that satisfy the algebraic condition~\eqref{eq:resonance_condition}.
For larger systems of vortices, or for systems including multiply-quantized vortices with larger integer-valued topological charges, the sets of resonant monomials will be different, and different types of terms will appear. The effective mass of each vortex depends on the number of particles of the $b$ component it traps. In the present analysis, we have assumed that the density of $b$ is sufficiently large, and the trapping dynamics sufficiently uniform, that each vortex captures essentially the same number of $b$-particles, resulting in identical effective masses for all vortices. This assumption is natural in regimes where the ``doping'' component is not too dilute and where the vortex cores represent energetically similar trapping sites. However, in a two-component Bose--Einstein condensate, this need not be the case: when the $b$ component is very dilute, different vortices may trap different numbers of particles, so that the vortex effective masses become non-uniform~\cite{Bellettini2024PRA} (and may effectively take non-integer values in units of the particle mass). In that scenario, there may exist no positive-integer-valued solutions to the resonance condition. Formally, this would imply that all terms appearing on the right-hand side of Eq.~\eqref{eq:homological_complete} can be solved; nevertheless, there will exist exponents $\balpha$ and $\bbeta$ for which the denominator of Eq.~\eqref{eq:homological_solution} becomes arbitrarily small. Such \emph{near resonances} produce \emph{small-divisor problems}, introducing additional technical difficulties in applying the perturbation method~\cite{SaVeMu2007}.

More broadly, the issue of whether vortex masses are equal or heterogeneous depends on the physical realization of the massive point-vortex model. While two-component condensates can naturally lead to unequal core masses due to statistical fluctuations in the amount of trapped foreign material, there exist other systems where the equal-mass assumption is more justified. For instance, in fermionic superfluids (and similarly in superconductors), vortex cores can be spontaneously filled by an intrinsic ``normal component'' generated by pair-breaking mechanisms~\cite{Caroli1964, Duan1992, Simula2018, Richaud2025Fermi, Levrouw2025}. In that case, the vortex core mass is not produced by externally introduced particles. Instead, it emerges as an intrinsic microscopic property of the vortex itself, making it far more reasonable to assume that all vortices carry the same effective core mass. Finally, in superfluid liquid helium, vortices can be doped with tracer particles~\cite{Bewley2006, Griffin2020, Peretti2023, Tang2023}, and the doping process can, in principle, be engineered to produce vortices with nearly identical tracer loads, again supporting the equal-mass approximation.

Expected applications of this improved and more accurate theoretical framework include the instabilities of many-vortex systems~\cite{Hernandez2024, Caldara2024, Simjanovski2025} and of multiply quantized vortices~\cite{Cuomo2024}, the formation~\cite{Tsuzuki2023} and sudden rearrangements of vortex crystals~\cite{Campbell1979, Coddington2004, Poli2023}, the normal modes of vortex necklaces~\cite{Chaika2023, Caldara2023, Caldara2024, Hernandez2024, Banger2025}, as well as the interaction of vortices in inhomogeneous media~\cite{Shinn2025} and curved substrates~\cite{Caracanhas2022}, or their interplay with other nonlinear excitations~\cite{Skogvoll2023, Edmonds2023}.   

\section*{Acknowledgments}
T.O. was supported by NSF grant DMS-2006736.
A.R. received funding from the European Union’s
Horizon research and innovation programme under the Marie Skłodowska-Curie grant agreement \textit{Vortexons} (no.~101062887), and from ``La Caixa'' Foundation  (ID 100010434) through the Junior Leader Fellowship LCF/BQ/PR25/12110014. A.R. further acknowledges support by the Spanish Ministerio de Ciencia e Innovación (MCIN/AEI/10.13039/501100011033, grant PID2023-147469NB-C21), by the Generalitat de Catalunya (grant 2021 SGR 01411), and by the ICREA {\it Academia}\ program.
A.R. is grateful to Pietro Massignan for helpful discussions.
R.G. was supported by NSF grant DMS-220616.

\appendix
\section*{Appendix: Proof of~\eqref{eq:integral_estimate}}
This proof essentially reproduces the proof of \cite[Lemma~2.8.2]{SaVeMu2007}.
We include it here for completeness.
Set
\begin{equation*}
  \mathbf{E}(\bw,\tau) \defeq \mathbf{F}(\bw,\tau) - \overline{\mathbf{F}}(\bw).
\end{equation*}
Then $\tau \mapsto \mathbf{E}(\bw,\tau)$ is $2\pi$-periodic with zero average.
Now, setting $\tau_{i} \defeq 2\pi i$ for $i \in \Z$, for every $T \in [0,t_{1}/\eps]$, there exists $m \in \N$ such that $\tau_{m} \le T < \tau_{m+1}$.
Therefore,
\begin{equation*}
  \int_{0}^{T} \mathbf{E}(\overline{\bw}(\tau), \tau)\,d\tau
  = \sum_{i=0}^{m-1} \int_{\tau_{i}}^{\tau_{i+1}} \mathbf{E}(\overline{\bw}(\tau), \tau)\,d\tau
  +  \int_{\tau_{m}}^{T} \mathbf{E}(\overline{\bw}(\tau), \tau)\,d\tau,
\end{equation*}

We have the following estimate for the first $m$ integrals:
For every $i \in \{0, \dots, m-1\}$, we have
\begin{align*}
  \norm{ \int_{\tau_{i}}^{\tau_{i+1}} \mathbf{E}(\overline{\bw}(\tau), \tau)\,d\tau }
  &= \norm{ \int_{\tau_{i}}^{\tau_{i+1}} \parentheses{ \mathbf{E}(\overline{\bw}(\tau), \tau) - \mathbf{E}(\overline{\bw}(\tau_{i}), \tau) } d\tau } \\
  &\le \int_{\tau_{i}}^{\tau_{i+1}} \norm{ \mathbf{E}(\overline{\bw}(\tau), \tau) - \mathbf{E}(\overline{\bw}(\tau_{i}), \tau) } d\tau \\
  &\le \lambda_{2} \int_{\tau_{i}}^{\tau_{i+1}} \norm{ \overline{\bw}(\tau)- \overline{\bw}(\tau_{i}) } d\tau \\
  &\le 2\pi\,\lambda_{2}\,c_{0}\,\eps.
\end{align*}
The first line is due to the vanishing average of $\mathbf{E}$; the third line follows because $\bw \mapsto \mathbf{E}(\bw,\tau)$ is Lipschitz with Lipschitz constant $\lambda_{2} > 0$ due to the assumptions; the fourth line follows by integrating~\eqref{eq:massless-fast_time}:
\begin{equation*}
  \norm{ \overline{\bw}(\tau)- \overline{\bw}(\tau_{i}) } \le \eps \int_{\tau_{i}}^{\tau} \norm{ \mathbf{G}(\overline{\bw}(\sigma)) } d\sigma
  \le c_{0}\,\eps
\end{equation*}
with some constant $c_{0} > 0$.

The norm of the last integral can also be bounded by some constant $k_{0} > 0$.
Hence we have
\begin{equation*}
  \norm{ \int_{0}^{T} \mathbf{E}(\overline{\bw}(\tau), \tau)\,d\tau } \le \underbrace{2\pi m}_{\tau_{m}}\,\lambda_{2}\,c_{0}\,\eps + k_{0} 
\end{equation*}
However, since $\tau_{m} \le T \le t_{1}/\eps$, we obtain
\begin{equation*}
  \norm{ \int_{0}^{T} \mathbf{E}(\overline{\bw}(\tau), \tau)\,d\tau } \le \lambda_{2}\,c_{0}\,t_{1} + k_{0} \eqdef c_{1}
  \quad
  \forall T \in [0, t_{1}/\eps].
\end{equation*}

\bibliography{Small-Mass_Point_Vortices}
\bibliographystyle{abbrvnat}

\end{document}